\documentclass[12pt,a4paper]{article}
\usepackage{amsthm}
\usepackage{amssymb}
\usepackage{latexsym}
\usepackage{xspace}
\usepackage{bbm}
\usepackage{amssymb}
\usepackage{amsmath}
\usepackage{graphicx}
\usepackage{epic}
\usepackage{eepic}
\usepackage{multirow}
\usepackage{gastex}
\usepackage{times}
\usepackage{thmtools, thm-restate}
\usepackage{mathtools}
\usepackage{MnSymbol}
\usepackage{enumitem}
\usepackage{color}
\usepackage{bbm}
\usepackage[hidelinks,bookmarks=false]{hyperref}

\declaretheorem{theorem}
\declaretheorem{lemma}
\declaretheorem[sibling=lemma]{proposition}
\declaretheorem[sibling=lemma]{definition}
\declaretheorem[sibling=lemma]{corollary}
\newtheorem*{remark}{Remark}
\newtheorem*{example}{Example}

\DeclareMathOperator{\E}{\mathbb{E}}
\DeclareMathOperator{\Var}{\mathrm{Var}}

\DeclareMathOperator{\rank}{rank}

\DeclareMathOperator{\im}{Im}

\newcommand{\TODO}[1]{}
\newcommand{\Cov}{\mathrm{Cov}}
\newcommand{\sen}{\mathrm{Sen}}

\newcommand{\gammin}{\gamma_{\min}}

\newcommand{\bits}[1]{\mathbb F_2^{#1}}
\newcommand{\evens}[1]{\mathbb V^{#1}}
\newcommand{\odds}[1]{\mathbb D^{#1}}

\newcommand\numberthis{\addtocounter{equation}{1}\tag{\theequation}}

\DeclarePairedDelimiter{\floor}{\lfloor}{\rfloor}

\AtBeginDocument{}

\title{On the weight distribution of random binary linear codes}
\author{
Nati Linial
\thanks{Department of Computer Science, Hebrew University, Jerusalem 9190401. e-mail: nati@cs.huji.ac.il. Supported by ERC grant 339096 "High-dimensional combinatorics".}
\and
Jonathan Mosheiff
\thanks{Department of Computer Science, Hebrew University, Jerusalem 9190401. e-mail: yonatanm@cs.huji.ac.il. Supported by the Adams Fellowship Program of the Israel Academy of Sciences and Humanities.}}
\date{}

\begin{document}
\maketitle

\begin{abstract}
We investigate the weight distribution of random binary linear codes. For $0<\lambda<1$ and $n\to\infty$ pick uniformly at random $\lambda n$ vectors in $\mathbb{F}_2^n$ and let $C \le \mathbb{F}_2^n$ be the orthogonal complement of their span. Given $0<\gamma<1/2$ with $0< \lambda < h(\gamma)$ let $X$ be the random variable that counts the number of words in $C$ of Hamming weight $\gamma n$. In this paper we determine the asymptotics of the moments of $X$ of all orders $o(\frac{n}{\log n})$.
\end{abstract}

\section{Introduction}
\label{sec:Intro}

Random linear codes play a major role in the theory of error correcting codes, and are also important in other areas such as information theory, theoretical computer science and cryptography \cite{McE78, Reg09, BMvT78, BJMM12}. Nevertheless, not much seems to be known about their properties. As already shown in Shannon's foundational paper \cite{Sha48}, random linear codes occupy a particularly prominent position in coding theory, being in some sense the best error correcting codes. The present paper is motivated by the contrast between the importance of random codes and the lack of our understanding. Our main aim is to improve our comprehension of the weight distribution of random binary linear codes.

The two most basic parameters of a code $C \subseteq \bits{n}$ are its {\em rate} $R = \frac{\log_2 |C|}{n}$ and its {\em relative distance} $\delta = \frac{\min \{\|x-y\| \;\mid\; x,y\in C\; x\ne y\}}{n}$, where $\|\cdot\|$ is the Hamming norm. Clearly, the rate of a $d$-dimensional linear code $C \subseteq \bits{n}$ is $\frac{d}{n}$, and its relative distance is $\frac{\min \{\|w\| \;\mid\; w\in C\; w\ne 0\}}{n}$.

It is a major challenge to understand the trade-off between rate and distance for linear as well as general codes. Concretely, given $0<\delta<\frac{1}{2}$, we wish to know the value of $\limsup R(C)$ where the the $\limsup$ is taken over all binary codes of relative distance at least $\delta$. The Gilbert-Varshamov lower bound (e.g., \cite{Gur10}, p. 82) states that $R\ge 1-h(\delta)$ is achievable, where $h$ is the binary entropy function. Despite many attempts, this bound has not been improved, nor shown to be tight, through over 60 years of intense investigations. The best known upper bound, from 1977, is due to McEliece, Rodemich, Rumsey and Welch \cite{MRRW77}. An alternative proof of this bound, using harmonic analysis on $\mathbb{F}_2^n$, was given in 2007 by Navon and Samorodnitsky \cite{NS07}. Note that this is an upper bound on {\em all} codes. It remains a major open question whether there are stricter upper bounds that apply only to {\em linear} codes. 

This paper concerns the weight distribution of random linear codes. Concretely, fix two rational numbers $0 <\gamma < \frac 12$ and $0<\lambda<h(\gamma)$, and let $n\in \mathbb N$ be such that $\lambda n$ is an integer and $\gamma n$ is an even integerֿ\footnote{For other ranges of the problem - See our Discussion.}. Let $C = C_{n,\lambda}$ be a random subspace of $\bits{n}$ that is defined via $C:=\{x\in\bits{n}|Kx=0\}$ where $K$ is a uniformly random $\lambda n\times n$ binary matrix. Clearly $\dim C \ge (1-\lambda)n$, and with very high probability equality holds. Denote $L = L_{n,\gamma} = \{x \in \bits{n}\mid \|x\|=\gamma n\}$. We investigate the distribution of the random variable $X = X_{n,\gamma,\lambda}= |C\cap L|$ for fixed $\gamma$ and $\lambda$ when $n\to \infty$. Clearly $\E(X) = N^{-\lambda}\binom{n}{\gamma n} = N^{h(\gamma)-\lambda+o(1)}$, where $N=2^n$. This follows since every $x\in L_{n,\gamma}$ belongs to a random $C_{n,\lambda}$ with probability $N^{-\lambda}$. Also, $\lim_{n\to\infty}\E(X) = \infty$, since, by assumption $\lambda<h(\gamma)$. 

It is instructive to compare what happens if rather than a random linear code $C$, we consider a uniformly random subset $C'\subset \bits{n}$, where every vector in $\bits{n}$ independently belongs to $C'$ with probability $N^{-\lambda}$. In analogy, we define $X' = |C'\cap L|$, and the distribution of $X'$ is clearly approximately normal. It would not be unreasonable to guess that $X$ behaves similarly, and in particular that its limit distribution, as $n\to\infty$ is normal.
However, as we show, the code's linear structure has a rather strong effect. Indeed $X$ does not converge to a normal random variable, and moreover, only a few of its central moments are bounded.

\subsection{Rough outline of the proof}

We seek to approximate the central $k$-th moments of $X$ for all $k\le o(\frac n{\log n})$. In Section \ref{sec:ReductionToEnumeration} we reduce this question to an enumeration problem that we describe next. We say that a linear subspace $U\le \bits {k}$ is {\em robust} if every system of linear equations that defines it involves all $k$ coordinates. Given a subspace $U\le \bits {k}$, let $T_U$ be the set of all $k\times n$ binary matrices where every column is a vector in $U$ and every row has weight $\gamma n$. We show that
\begin{equation} \label{eqn:MomentIntro}
\E\left((X-\E(X))^k\right) = \Theta\left(\sum_{d=0}^{k-1}N^{-\lambda d} \sum_{\substack{V\le \bits{k}\\\dim(V)=d\\V\text{ robust}}}|T_V|\right).
\end{equation}
The main challenge is to estimate the internal sum, but understanding the interaction with the outer sum is nontrivial either. The reason that we can resolve this problem is that the main contributors to the internal sum are fairly easy to describe. As it turns out, this yields a satisfactory answer even though we provide a rather crude upper bound on all the other terms. 

A key player in this story is the space of even-weight vectors $V=\evens{k} \le \bits{k}$.
In Section \ref{sec:MatrixEnumerationEvens} we solve this enumeration problem for this space, and show that $|T_{\evens{k}}|  \approx N^{F(k,\gamma)}$ up to a factor that is polynomial in $n$ and exponential in $k$. Here $F(k,\gamma)$ is the entropy of a certain entropy maximizing probability distribution on $\evens{k}$. In our proof, we generate a $k\times n$ matrix $A$ with i.i.d.\ columns sampled from this distribution, and compute the probability that $A\in T_{\evens{k}}$. The function $F$ has the explicit description
$$F(k,\gamma) = \min_{1>x>0} \log_2\left((1+x)^k+(1-x)^k\right)-k\gamma\log_2 x - 1$$
and its asymptotic behavior for large $k$ is:
$$F(k,\gamma) = kh(\gamma)-1 + O((1-2\gamma)^k).$$

In Section \ref{sec:MatrixEnumerationGeneral} we use the result of Section \ref{sec:MatrixEnumerationEvens} to bound $|T_U|$ for a general robust $U\le \bits{k}$. Consider a robust space $U\le \bits{k}$ of the form $\bigoplus_{i=1}^c \evens{m_i}$, where $\sum m_i = k$. Clearly, $|T_U| = \prod_{i=1}^c{|T_{\evens{m_i}}|}\approx N^{\sum_{i=1}^c F(m_i,\gamma)}$. Hence, finding a space of this form of given dimension that maximizes $|T_U|$ translates into a question about the dependence of $F(m,\gamma)$ on $m$. We show (Lemma \ref{lem:FConvexInk}) that this function is convex, so that the optimum is attained at $m_1 = k-2c+2$ and $m_2=m_3=\ldots=m_c=2$. 

We show that if $U\le \bits{k}$ is robust and not a product of {\em Even} spaces, then there is some $V$ of this form and of the same dimension with $|T_V|\ge |T_U|$. We reduce the proof of this claim (Equation \ref{eqn:T_VBound}) to the analysis of $m\times n$ matrices where every row weighs $\gamma n$, the first $\delta n$ columns have odd weight and the last $(1-\delta) n$ ones are even. A key step in the proof (Lemma \ref{lem:FMonotoneInDelta}) shows that the number of such matrices decreases with $\delta$.

Finally, in Section \ref{sec:DeriveMoments}, the results of the previous sections are put together to find the dominating terms of Equation \ref{eqn:MomentIntro}, yielding the moments of $X$. For even $k$, we show that the dominating terms are those corresponding to either $d=\frac k2$ or $d = k-1$, and respectively, to the subspaces $\bigoplus_{i=1}^{k/2}\evens{2}$ or $\evens{k}$. More precisely, there exists some $k_0(\gamma, \delta)$ such that the former dominates when $k \le k_0$ and the latter when $k > k_0$. The behavior of odd order moments is similar, although slightly more complicated to state.

Theorems \ref{thm:mainEven} and \ref{thm:mainOdd} in Section \ref{sec:DeriveMoments}, deal with even and odd order moments, respectively. Theorem \ref{thm:mainNormalized} gives the central moments of the normalized variable $\frac {X}{\sqrt{\Var(X)}}$.
\begin{restatable}{theorem}{MainNormalized} \label{thm:mainNormalized}
Fix $\gamma <\frac 12$ and $0<\lambda<h(\gamma)$ and let 
$$k_0 = \min \left\{m\mid F(m,\gamma)-(m-1)\lambda > \frac m2(h(\gamma)-\lambda) \right\}.$$
Then, for $2 \le k \le o(\frac n{\log n})$,

$$
\frac{\E(X-\E(X)^k)}{\Var(X)^{\frac k2}} = \begin{cases}
o(1) &\text{if }k\text{ is odd and }< k_0\\
(1+o(1))\cdot k!! &\text{if }k\text{ is even and }< k_0\\
N^{F(m,\gamma)-\frac k2h(\gamma)-(\frac k2-1)\lambda -\frac{k\log n}{4n} +O(\frac kn)}&\text{if }k\ge k_0\\
\end{cases}
$$
\end{restatable}

We call the reader's attention to the following interesting point on which we elaborate below. For fixed $\gamma$ and $\lambda$ there is a bounded number of moments for which our distribution behaves as if it were normal, but from that index on its linear structure starts to dominate the picture and the moments become unbounded. (See Figure \ref{fig:k0Ranges}).

\begin{figure}
\centering
\includegraphics[width=0.95\textwidth]{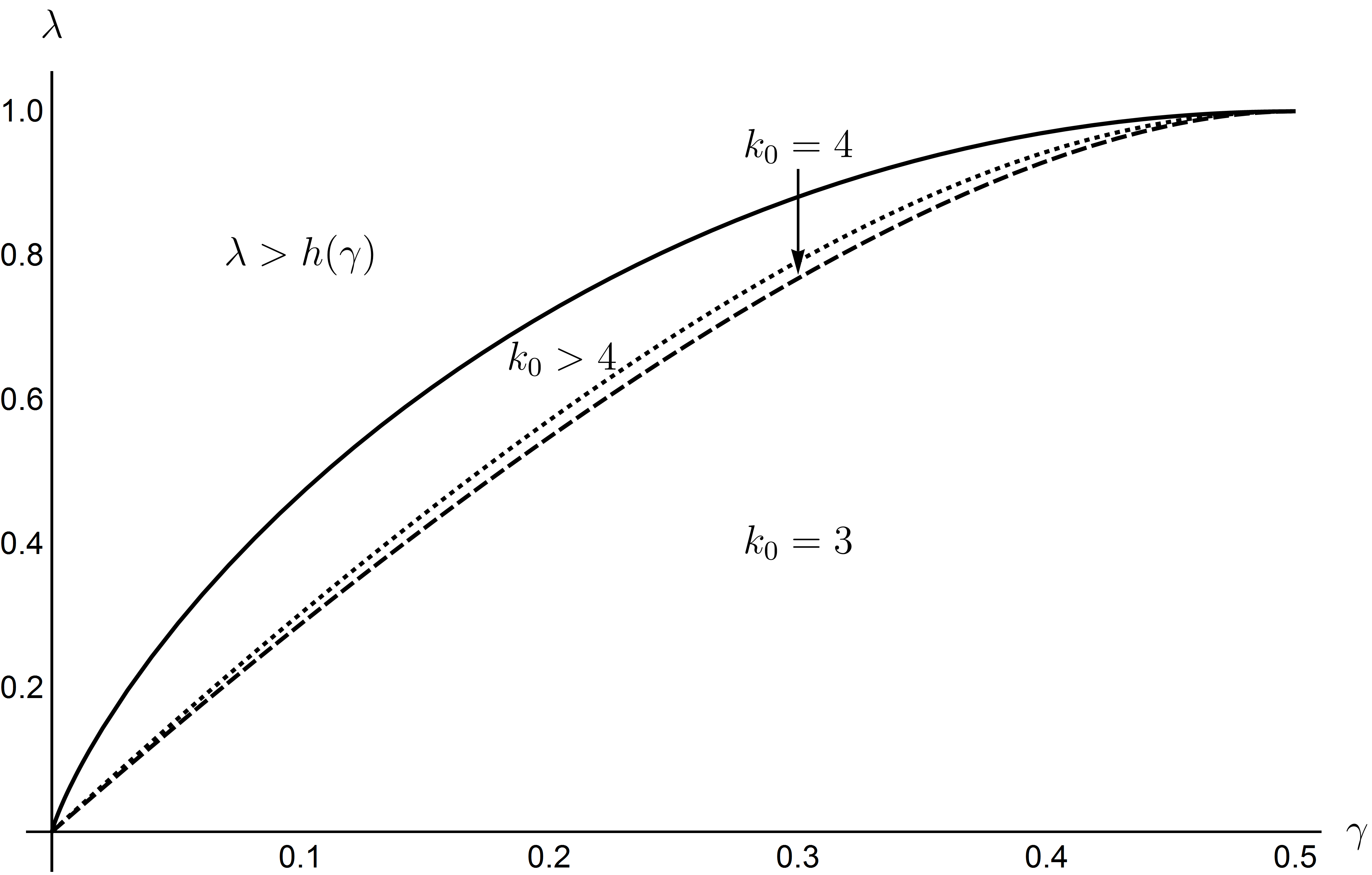}
\caption{\label{fig:k0Ranges}Illustration for Theorem \ref{thm:mainNormalized}. For $k < k_0= k_0(\gamma,\lambda)$ the $k$-th moment of $X$ is that of a normal distribution. The relevant range $\lambda < h(\gamma)$ is below the solid line. Note that $k_0=3$ for much of the parameters range.}
\end{figure}

\subsection{Preliminaries} 

{\bf General:}
Unless stated otherwise, all logarithms here are to base $2$.\\
Our default is that an asymptotic statement refers to $n\to \infty$, while the parameters $\gamma$ and $\lambda$ take fixed arbitrary values within their respective domains. Other parameters such as $k$ may or may not depend on $n$.\\
{\bf Entropy:} We use the standard notation $h(t)=-t\cdot\log t-(1-t)\cdot\log(1-t)$. Entropy and conditional entropy are always binary.\\
{\bf Linear algebra:} $U\le V$ means that $U$ is a linear subspace of the vector space $V$. The {\em weight}, $\|u\|$ of a vector $u\in \bits{n}$ is the number of its $1$ coordinates. Accordingly we call $u$ even or odd. Likewise, the weight $\|A\|$ of a binary matrix $A$, is the number of its $1$ entries. \\
The sets of even and odd vectors in $\bits{n}$ are denoted by $\evens{n}$ and $\odds{n}$.\\
The $i$-th row of a matrix $A$ is denoted by $A_i$. If $I\subseteq[k]$ then $A_I$ is the sub-matrix consisting of the rows $\{A_i \mid i\in I\}$. Also $v_I$ is the restriction of the vector $v$ to the coordinates in $I$.\\
For a subspace $U\le \bits{k}$ and $I\subset [k]$ we denote by $U_I$ the projection of $U$ to the coordinates in $I$, i.e., $U_I = \{u_I\mid u\in U\}$, and we use the shorthand $d_I(U) = \dim U_I$, and $d(U)=\dim U$. 

\section{From moments to enumeration.} \label{sec:ReductionToEnumeration}
To recap: $C=C_{n,\lambda}$ is a random linear subspace of $\bits{n}$, and $L=L_{n,\gamma}$ is the $\gamma n$-th layer of $\bits{n}$. We fix $0<\gamma<1$, $0<\lambda<h(\gamma)$, so that $\lambda n$ is an integer and $\gamma n$ is an even integer, and we start to investigate the moments of $X = |C\cap L|$, as $n\to\infty$.

The probability that $C$ contains a given subset of $\bits{n}$ depends only on its linear dimension:

\begin{proposition} \label{prop:CContainsSmallSet}
If $Y\subseteq \bits{n}$ has dimension $\dim(Y)=d$, then $\Pr(Y\subseteq C) = N^{-\lambda d}$.
\end{proposition}
\begin{proof}
As mentioned, we think of $C$ as the kernel of a uniform random $\lambda n\times n$ binary matrix $K$, so $Y \subseteq C$ iff every row of $K$ is orthogonal to $Y$. The probability of this event is $2^{-d}$ for a given row, and $2^{-\lambda n d}=N^{-\lambda d}$ for all rows together.
\end{proof}

\subsection{Interpreting the central moments of $X$}
We turn to express $X$ and its moments in terms of indicator random variables.

\begin{definition}
For a vector $u\in \bits{n}$, let $Y_u$ be the indicator for the event that $u\in C$. For a binary $k\times n$ matrix $A$ we let $Y_A$ be the indicator random variable for the event that every row of $A$ is in $C$.
\end{definition}

Proposition \ref{prop:CContainsSmallSet} plainly yields the first two central moments of $X$.
$$\E(X) = \sum_{u\in L}\E(Y_u) = |L|N^{-\lambda} = \binom{n}{\gamma n} N^{-\lambda} = N^{h(\gamma)-\lambda-\frac{\log n}{2n}+O(\frac 1 n)}.$$ 
Proposition \ref{prop:CContainsSmallSet} also implies that $\Cov(Y_u,Y_v) = 0$ for every $u\ne v\in L$. Hence,
$$\Var(X) = \sum_{u\in L}\Var(Y_u) = \binom{n}{\gamma n}N^{-\lambda}(1-N^{-\lambda}) = N^{h(\gamma)-\lambda-\frac{\log n}{2n}+O(\frac 1 n)}.
$$

In words, the first two moments of $X$ are not affected by the linearity of $C$. 

We now turn to higher order moments. Specifically we wish to compute the $k$-th central moment of $X$ for any $2<k \le o(\frac{n}{\log n})$.
 
We denote by $W_k=W_{k, \gamma}$ the set of binary $k\times n$ matrices in which every row has weight $\gamma n$. We also introduce
\begin{definition}
For a subspace $U\le \bits{k}$ we denote 
$$T_{U,n,\gamma} = T_U = \{A\in W_k \mid \im A \subseteq U\}$$
and
$$\overline T_{U,n,\gamma} = \overline T_U = \{A\in W_k \mid \im A = U\}.$$
\end{definition}

Let us expand the $k$-th central moment.
\begin{align*} \label{eqn:k'thMomentInitial}
\E\left((X-\E(X))^k\right) &= \E\left(\left(\sum_{u\in L}Y_u-\sum_{u\in L}\E(Y_u)\right)^k\right)\\&=\sum_{u_1,\ldots,u_k\in L} \sum_{I\subseteq [k]}\E\left(\prod_{i\in I}Y_{u_i}\right)\prod_{j\in [k]\setminus I}\left(-\E\left(Y_{u_j}\right)\right). \numberthis
\end{align*}
If $A$ is the matrix with rows $u_1,\ldots,u_k$, then by Proposition \ref{prop:CContainsSmallSet} this equals
\[\sum_{A\in W_{k}} \sum_{I\subseteq [k]}(-1)^{k-|I|}\cdot N^{-\lambda\cdot (\rank A_I-k+|I|)}.
\]
We group the matrices $A\in W_k$ with the same image $U$ and rewrite the above as
\[
\sum_{U\le \bits{k}} |\overline T_U| \sum_{I\subseteq [k]} (-1)^{k-|I|}\cdot N^{-\lambda\cdot (d_I(U)-k+|I|)},
\]
which we restate as
\begin{equation}\label{eqn:k'thMoment}
\E\left((X-\E(X))^k\right) = \sum_{U\le\bits{k}}|\overline T_U|R_U,
\end{equation}
where for any $U\le \bits{k}$
\begin{equation}\label{eqn:R_UDefinition}
R_U = \sum_{I\subseteq [k]}(-1)^{k-|I|}\cdot N^{-\lambda\cdot (d_I(U)-k+|I|)}
\end{equation} 

We proceed as follows:
\begin{enumerate}
\item We recall the notion of a {\em robust} linear subspace of $\bits{k}$, and compute $R_U$ separately for robust and non-robust subspaces.
\item Using M\"obius inversion, we restate Equation \ref{eqn:k'thMoment} in terms of $|T_U|$ rather than $|\overline T_U|$.
\end{enumerate}

\subsubsection{Computing $R_U$}
It is revealing to consider our treatment of $X$ alongside a proof of the Central Limit Theorem (CLT) based on the moments method (e.g., \cite{Fil10}). In that proof, the $k$-th moment of a sum of random variables of expectation zero is expressed as a sum of expectations of degree-$k$ monomials, just as in our Equation \ref{eqn:k'thMomentInitial}. These monomials are then grouped according to the relations between their factors. In the CLT proof, it is assumed that each tuple's non-repeating factors are independent, so monomials are grouped according to their degree sequence. Here, and specifically in Equation \ref{eqn:k'thMoment}, we need a more refined analysis that accounts for the linear matroid that is defined by the monomial's factors. 

In the proof the the CLT there holds $\E(M) = 0$ for every monomial $M$ that contains a degree-$1$ factor $Y$. This follows, since $\E(Y)=0$ and the rest of the monomial is independent of $Y$. Something similar happens here too. If $u$ does not participate in any linear relation with the other factors in its monomial, then $Y_u$ can play a role analogous to that of $Y$. This intuition is captured by the following definition and proposition.

\begin{definition}
Let $U\le\bits{k}$ be a linear subspace. We say that its $i$-th coordinate is {\em sensitive} if $d_{[k]\setminus \{i\}}(U)=d(U)-1$. We denote by $\sen(U)$ the set of $U$'s sensitive coordinates. Also, if $\sen(U)=\emptyset$, we say that $U$ is {\em robust}.
\end{definition}
It is not hard to see that equivalently, robustness means that every $1$-co-dimensional coordinate-wise projection of $U$ has the same dimension as $U$. Yet another description is that every system of linear equations that defines $U$ must involve all coordinates.

\begin{proposition}\label{prop:R_UAsymptotics}
For $U\le \bits{k}$ it holds that
\begin{enumerate}
\item If $U$ is robust then $R_U = \Theta\left(N^{- d(U) \lambda }\right)$.
\item If $U$ is not robust then $R_U = 0$.
\end{enumerate}
\end{proposition}
\begin{proof}
We use here the shorthand $d = d(U)$ and $d_I = d_I(U)$.

We start with the case of a robust $U$. Note that for every $I\subsetneq [k]$ there holds $d_I\ge d -k+|I|+1$. For let us carry out the projection as \mbox{$k-|I|$} steps of $1$-co-dimensional projections. At each step the dimension either stays or goes down by one. But since $U$ is robust, in the first step the dimension stays.

We claim that in the expression for $R_U$ in Equation \ref{eqn:R_UDefinition}, the term $N^{-\lambda d}$ that corresponds to $I=[k]$ dominates the rest of the sum. Indeed, each of the other $2^k-1$ summands is $\pm\Theta(N^{-\lambda(d+1)})$. Consequently, $R_U = \Theta(N^{-\lambda d})$.

Let us consider next a non-robust $U$. If $I$ is a set of sensitive coordinates and $J$ is a set of non-sensitive coordinates, then $d_{I\cup J}=|I|+d_{J}$. Consequently:
\begin{align*}R_U&=  \sum_{I\subseteq \sen(U)} \sum_{J\subseteq [k]\setminus\sen(U)}  (-1)^{k-|I|-|J|} N^{-\lambda (d_{I\cup J}+k-|I|-|J|)} \\ &=  \sum_{I\subseteq \sen(U)}\sum_{J\subseteq [k]\setminus\sen(U)}  (-1)^{k-|I|-|J|} N^{-\lambda (d_{J}+k-|J|)} \\
&= \left(\sum_{I\subseteq \sen(U)}(-1)^{|I|}\right)\left(\sum_{J\subseteq [k]\setminus\sen(U)}  (-1)^{k-|J|} N^{-\lambda (d_{J}+k-|J|)}\right) = 0.
\end{align*}
\end{proof}

\subsubsection{From $|\overline T_U|$ to $|T_U|$}

In order for Equation \ref{eqn:k'thMoment} to be expressed in terms of $|T_U|$ rather than $|\overline T_U|$ we can appeal to the M\"obius inversion formula for vector spaces over a finite field (e.g., \cite{Sta11}, Ch 3.10).
\begin{align*}
\E\left((X-\E(X))^k\right) &= \sum_{U\le\bits{k}}R_U \sum_{V\le U} (-1)^{d(U)-d(V)}\cdot2^{\binom{d(U)-d(V)}2}|T_V| \\
&=\sum_{V\le \bits{k}}|T_V|\sum_{V\le U\le \bits{k}}R_U (-1)^{d(U)-d(V)}\cdot2^{\binom{d(U)-d(V)}2}. 
\end{align*}

Grouping the $U$'s by their dimension $i=d(U)$, we express the above as
$$
\sum_{V\le \bits{k}}|T_V| (-1)^{d(V)}\sum_{i=d(V)}^k (-1)^{i}\cdot 2^{\binom{i-d(V)}2} \sum_{\substack{V\le U\le \bits{k}\\ d(U) = i}} R_U.
$$
By Proposition \ref{prop:R_UAsymptotics}, this sum can be further rewritten as
$$
\sum_{V\le \bits{k}}|T_V| (-1)^{d(V)}\sum_{i=d(V)}^k (-1)^{i}\cdot 2^{\binom{i-d(V)}2} \cdot N^{-\lambda i} \cdot Z_{i, V}
$$
where
$$
Z_{i, V}=|\{U \mid V\le U\le \bits{k} ~\wedge~ d(U)=i ~\wedge~ U\text{~is robust}\}|.
$$

Note that if $V$ is non-robust then every $U\ge V$ is also non-robust. Hence, the outer sum terms corresponding to non-robust $V$'s vanish. If $V$ is robust, we claim that the inner sum is dominated by the term $i=d(V)$ and that consequently 
\begin{equation}\label{eqn:k'thMomentAfterMoebius}
\E\left((X-\E(X))^k\right) = \Theta\left(\sum_{\substack{V\le \bits{k}\\V\text{ robust}}}|T_V|\cdot N^{-\lambda d(V)}\right).
\end{equation}

Indeed, $V$ is contained in at most $2^{2+(i-d(V))(k-i)}$ dimension-$i$ spaces, so the absolute value of the inner sum's $i$-term is at most
$$2^{\binom{i-d(V)}2 -\lambda n i + 2 +(i-d(V))(k-i)}=2^{2+(i-d(V))(k-\frac{i+d(V)+1}2)-i\lambda n}\le 2^{-i(\lambda n  +1 - k)+2}.$$	 

In order to proceed we need to estimate the cardinalities $|T_V|$. As we show in Sections \ref{sec:MatrixEnumerationEvens} and \ref{sec:MatrixEnumerationGeneral}, at least for large enough $k$, Equation {\ref{eqn:k'thMomentAfterMoebius}} is dominated by the term $V=\evens{k}$, the subspace of even-weight vectors.

\section  {The intersection of $\evens{k}$ and the $\gamma n$-th layer}
\label{sec:MatrixEnumerationEvens}
In this section we give tight estimates for $|T| = |T_{\evens{k},n,\gamma}|$. As usual we assume that $0 < \gamma < \frac 12$ and $\gamma n$ is an even integer. We need the following terminology:

\begin{definition}
Let $A_{k\times n}$ be a binary matrix.
\begin{itemize}
\item A row of $A$ is said to satisfy the {\em row condition} if it weighs $\gamma n$. If this holds for every row of $A$, we say that $A$ satisfies the row condition.
\item The {\em column condition} for $A$ is that every column be of even weight.
\item Recall that $T_{\evens{k},n,\gamma}$ is the set of $k\times n$ binary matrices satisfying both the row and the column conditions.
\end{itemize}
\end{definition}

To estimate $|T|$, we define a certain probability measure $\pi=\pi_{k,n,\gamma}$ on binary $k\times n$ matrices. Under this measure the probability of the event $T$ is not too small, viz., inverse polynomial in $n$ and exponentially small in $k$. We then estimate $|T|$ by using our bounds on this probability.

In this distribution $\pi$ columns are chosen independently according to a distribution $P=P_{k,\gamma}$ that is supported on $\evens{k}$, and is $S_k$-invariant. Naturally, we choose it so that for every $i$: 

\begin{equation} \label{eqn:pExpectation}
\Pr_{u\sim P}(u_i=1) = \gamma.
\end{equation}
We seek a distribution $P$ of largest possible entropy that satisfies these conditions. The intuition behind this choice has to do with the theory of exponential families (E.g., \cite{WJ08} Chapter 3) which provides a framework to describe and study maximum entropy distributions. However, we do not directly rely on this theory so that this paper remains self-contained.

Concretely, for some $1 > \alpha > 0$ and for every $u\in \evens{k}$ we define
\begin{equation} \label{eqn:pDefinition}
P(u) = \frac{\alpha ^{\|u\|}}{Z}
\end{equation}

Here $Z=Z(\alpha,k) = \sum_{u\in \evens{k}}\alpha^{\|u\|}$. We claim that there is a unique $1 > \alpha > 0$ for which Condition \ref{eqn:pExpectation} holds. First, note that 
$$Z = \sum_{w\text { is even}} \binom k w\alpha^w = \frac{(1+\alpha)^k+(1-\alpha)^k}2.$$
Also,
$$\Pr_{u\sim P}(u_i=1) = \sum_{w\text{ is even}}\frac{\binom{k-1}{w-1}\alpha^w}{Z}= \alpha\frac{(1+\alpha)^{k-1}-(1-\alpha)^{k-1}}{(1+\alpha)^k+(1-\alpha)^k}$$
so that Equation \ref{eqn:pExpectation} becomes
\begin{equation} \label{eqn:gamma(alpha)}
\alpha\frac{(1+\alpha)^{k-1}-(1-\alpha)^{k-1}}{(1+\alpha)^k+(1-\alpha)^k}=\gamma.
\end{equation}
Denote the left side of this expression by $\gamma(k,\alpha)$. 

\begin{proposition} \label{prop:gammaIncreases}
Let $k \ge 2$. In the range $0 < \alpha < 1$ the function $\gamma(k,\alpha)$  increases from $0$ to $\frac 12$.
\end{proposition}
\begin{proof}
In the following, the sums are over even $i$, $j$ and $t$:
$$
\frac{\partial \gamma(k,\alpha)}{\partial \alpha} = \frac{\left(\sum_{i}i\binom{k-1}{i-1}\alpha^{i-1}\right)\left(\sum_{j}\binom kj \alpha^{j}\right)-\left(\sum_{i}\binom{k-1}{i-1}\alpha^{i}\right)\left(\sum_{j}j\binom kj \alpha^{j-1}\right)}{Z^2}. 
$$
Denoting $t = j+i$, the above equals
$$
\frac{\sum_{t} \alpha^t \sum_{i}(2i-t)\binom {k-1}{i-1}\binom k{t-i}}{\alpha Z^2} = \frac{\sum_{t} \alpha^t \sum_{i}(2i-t)i\binom {k}i\binom k{t-i}}
{k\alpha Z^2}.
$$
Grouping the $i$ and $t-i$ terms of the inner sum yields
$$
\frac{\sum_{t} \alpha^t \sum_{i}(2i-t)^2\binom {k}i\binom k{t-i}} {2k\alpha Z^2},
$$
which is clearly positive.
\end{proof}

It follows that the function $\gamma=\gamma(k,\alpha)$ has an inverse with respect to $\alpha$, which we denote by $\alpha=\alpha(k,\gamma)$.

\begin{proposition}\label{prop:alphaAsymptotics}
$$\alpha(k,\gamma) = \frac{\gamma}{1-\gamma} + O((1-2\gamma)^k)$$
for every fixed $\gamma \in (0,\frac 12)$ and ${k\to \infty}$.
\end{proposition}
\begin{proof}
The proposition follows from the following inequality: 
$$\gamma\left(k,\frac{\gamma_0}{1-\gamma_0}\right) \le \gamma_0 \le \gamma\left(k,\frac{\gamma_0+\epsilon}{1-\gamma_0}\right)$$
where $\epsilon=2\gamma_0\cdot \frac{(1-2\gamma_0)^{k-1}}{1-(1-2\gamma_0)^{k-1}}$.

The lower bound is easily verified, since
$$
\gamma\left(k,\frac{\gamma_0}{1-\gamma_0}\right) = \gamma_0 \cdot\frac{1-(1-2\gamma_0)^{k-1}}{1+(1-2\gamma_0)^k}.
$$
For the upper bound, our claim,
$$
\gamma\left(k,\frac{\gamma_0+\epsilon}{1-\gamma_0}\right)= (\gamma_0+\epsilon)\frac{(1+\epsilon)^{k-1}-(1-2\gamma_0-\epsilon)^{k-1}}{(1+\epsilon)^k+(1-2\gamma_0-\epsilon)^k}\ge \gamma_0,
$$
is equivalent to
$$
(1+\epsilon)^{k-1}\epsilon \ge (2\gamma_0+\epsilon)(1-2\gamma_0-\epsilon)^{k-1}.
$$
To see this, note that the l.h.s.\ is $\ge \epsilon$, and the r.h.s.\ is $\le {(2\gamma_0+\epsilon)(1-2\gamma_0)^{k-1}}$. Finally, the latter two expressions are identical.

\end{proof}

We turn to compute the entropies of the distributions we have just defined:
$$h(\pi) = nh(P)$$
where
\begin{align*}
h(P) &= -\sum_{u\in\evens{k}}\frac{\alpha^{\|u\|}}{Z}\log\frac{\alpha^{\|u\|}}{Z}= \log Z\cdot\sum_{u\in \evens{k}} \frac{\alpha^{\|u\|}}{Z} -\sum_{u\in\evens{k}}\frac{\|u\|\alpha^{\|u\|}}{Z}\log\alpha\\
&= \log Z - \E_{u\sim P}(\|u\|)\log \alpha = \log Z - k\Pr_{u\sim P}(u_1=1)\log \alpha = \log Z -k\gamma\log \alpha.
\end{align*}
To sum up:
$$ 
h(\pi) = n(\log Z -k\gamma\log \alpha).
$$

We denote 
$$
F(k,\gamma) = \frac{h(\pi)}{n} = \log Z -k\gamma\log \alpha = \log((1+\alpha)^k+(1-\alpha)^k)-k\gamma\log\alpha-1.
$$
We next evaluate $\pi(A)$ for a matrix $A\in T$. Let $u_1,\ldots ,u_n$ be the columns of $A$. Then
$$
\pi(A) = \prod_{i=0}^n P(u_i) = \prod_{i=1}^{n} \frac{\alpha^{\|u_i\|}}{Z}= \frac{\alpha^{\|A\|}}{Z^{n}} = \frac{\alpha^{\gamma k n}}{Z^{n}} = 2^{-h(\pi)}.
$$
Since $\pi$ is constant on $T$, this yields an expression for $|T|$. Namely,
\begin{equation} 
|T| = \frac{\Pr_{A\sim \pi}(A\in T)}{\pi(A)} = \Pr_{A\sim \pi}(A\in T)\cdot 2^{h(\pi)}. \label{eqn:TFromEntropyAndProbability}
\end{equation}

This is complemented by the following Lemma.
\begin{restatable}{lemma}{lemTProbability}\label{lem:TProbability}
Fix $\gamma \in (0,\frac 12)$. Then, for every $k\ge 3$ and $n\in \mathbb N$, there holds
$$
\Pr_{A\sim \pi_{k,n,\gamma}}(A\in T) = n^{-\frac k2}\cdot 2^{\pm O(k)}.$$
\end{restatable}

We will prove Lemma \ref{lem:TProbability} at the end of this section. Before doing so, we wish to explore its implications. Together with Equation \ref{eqn:TFromEntropyAndProbability}, Lemma \ref{lem:TProbability} allows us to conclude that
\begin{equation} \label{eqn:TAsymptotics}
|T| = N^{F(k,\gamma)-\frac{k\log n}{2n}\pm O(\frac kn)}
\end{equation}
if $k \ge 3$. 

For $k=2$, a matrix in $|T|$ is defined by its first row, so
$$
|T| = \binom n{\gamma n} = N^{h(\gamma) - \frac {\log n}{2n} + O(\frac 1n)}.$$

As we show later, $F(k,\gamma)$ has a linear (in $k$) asymptote. Consequently, the exponents in Equation \ref{eqn:TAsymptotics} are dominated by the $F(k,\gamma)$ term. Thus, to understand $|T|$'s behavior we need to investigate $F$, which is what we do next. 

\subsection{Basic properties of $F(k,\gamma)$}
We start with several simple observations about $F(k,\gamma)$. 
\begin{proposition}
For $\gamma \in (0,\frac 12)$ there holds $F(2,\gamma) = h(\gamma)$. Also, $F(k,\gamma) \le k-1$ for all $k \ge 2$.
\end{proposition}
\begin{proof}
For the first claim, note that $\gamma(2,\alpha) = \frac{\alpha^2}{1+\alpha^2}$ so $\alpha(2,\gamma) = \left(\frac{\gamma}{1-\gamma}\right)^{\frac 12}$. Hence 
$$F(2,\gamma) = \log Z - 2\gamma \log \alpha = \log (1+\alpha^2) - \gamma\log (\alpha^2) = h(\gamma).$$

The second claim holds since $F(k,\gamma) = h(P)$ is the binary entropy of a distribution with support size $2^{k-1}$. 
\end{proof}

Next we develop an efficient method to calculate $F$ to desirable accuracy. We recall (e.g., \cite{Cov06}, p. 26) the notion {\em cross entropy} of $D, E$, two discrete probability distributions $H(D,E):=-\sum_{i}D(i)\log E(i)$. Recall also that $H(D,E)\ge h(D)$ with equality if and only if $D=E$. We apply this to $P=P_{k,\gamma}$, with $\alpha = \alpha(k,\gamma)$ and to $Q$, a distribution defined similarly according to Equation \ref{eqn:pDefinition}, but with some $x$ in place of $\alpha$. Then 
\begin{align*} \label{eqn:CrossEntropy}
F(k,\gamma)&=h(P)\le H(P\mid Q)  = -\sum_uP(u)\log{Q(u)} = -\sum_u P(u)\log \frac{x^{\|u\|}}{Z_k(x)} \\
&=\log Z_k(x) - \sum_u P(u)\|u\| \cdot \log(x) = \log Z_k(x) - \E_{u\sim P}(\|u\|) \cdot\log(x) 
\\ & =\log Z_k(x) - \gamma k\log(x) \numberthis
\end{align*}
Denote the r.h.s.\ of Equation \ref{eqn:CrossEntropy} by $g(k,\gamma,x)$. It follows that for an integer $k \ge 2$ and $\gamma\in(0,\frac 12)$,
\begin{equation} \label{eqn:FByMinimization}
F(k,\gamma) = \min_{x\in(0,1)} g(k,\gamma,x) = \min_{x\in(0,\infty)} \log \left((1+x)^k+(1-x)^k\right)-\gamma k \log(x)-1.
\end{equation}
This minimum is attained at $x=\alpha(k,\gamma)$. Note that this expression allows us
to conveniently compute $F$ to desirable accuracy (see Figure \ref{fig:g}). Also, we take Equation \ref{eqn:FByMinimization} as a definition for $F(k,\gamma)$ for all real positive $k$. 

\begin{figure}
    \centering
    \begin{minipage}{0.47\textwidth}
        \centering
        \includegraphics[width=\textwidth]{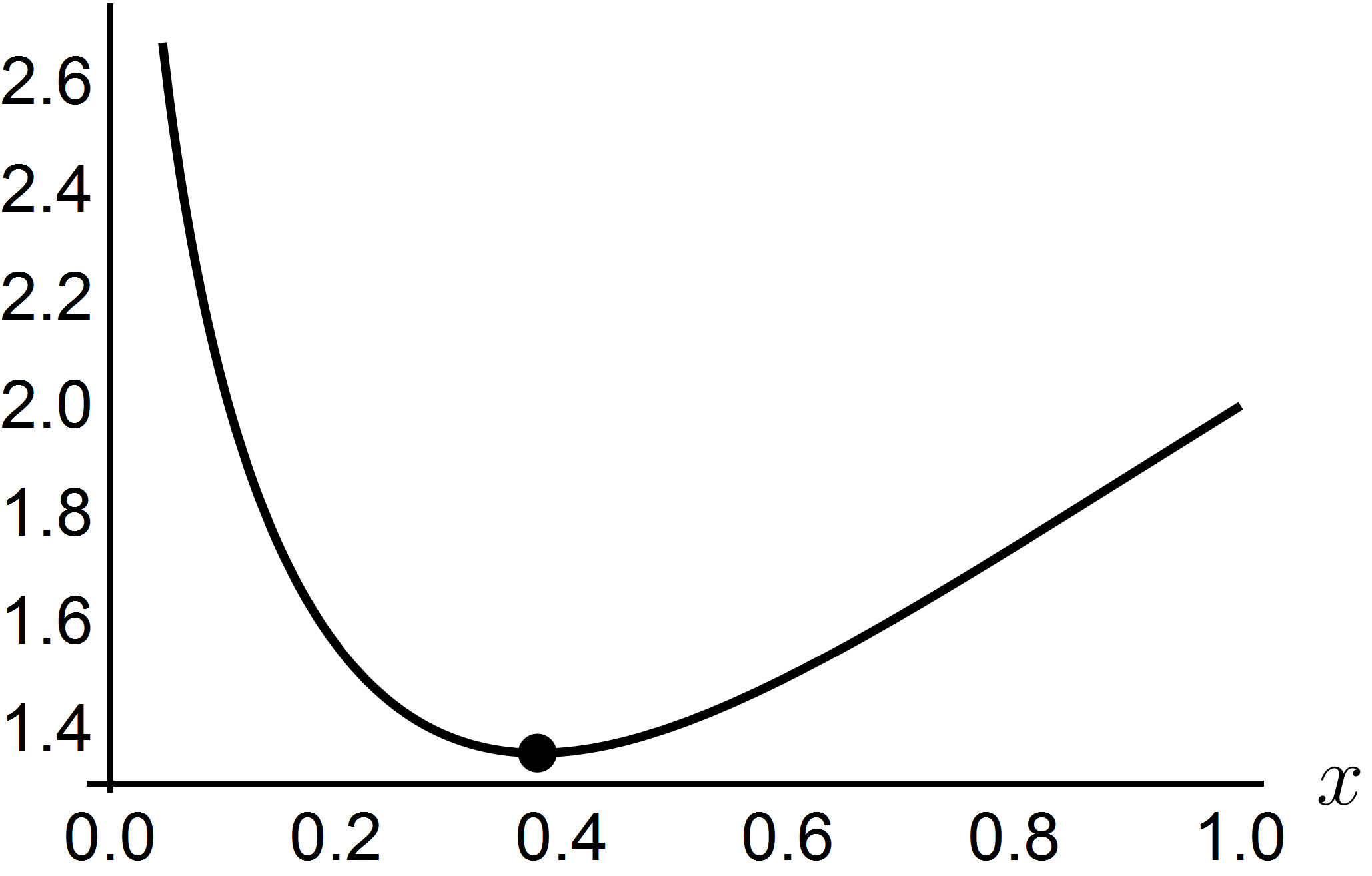}
        \caption{\label{fig:g}The function $g(3,\frac 15, x)$ and its minimum (see Equation \ref{eqn:FByMinimization}).}
    \end{minipage}\hfill
    \begin{minipage}{0.47\textwidth}
        \centering
        \includegraphics[width=\textwidth]{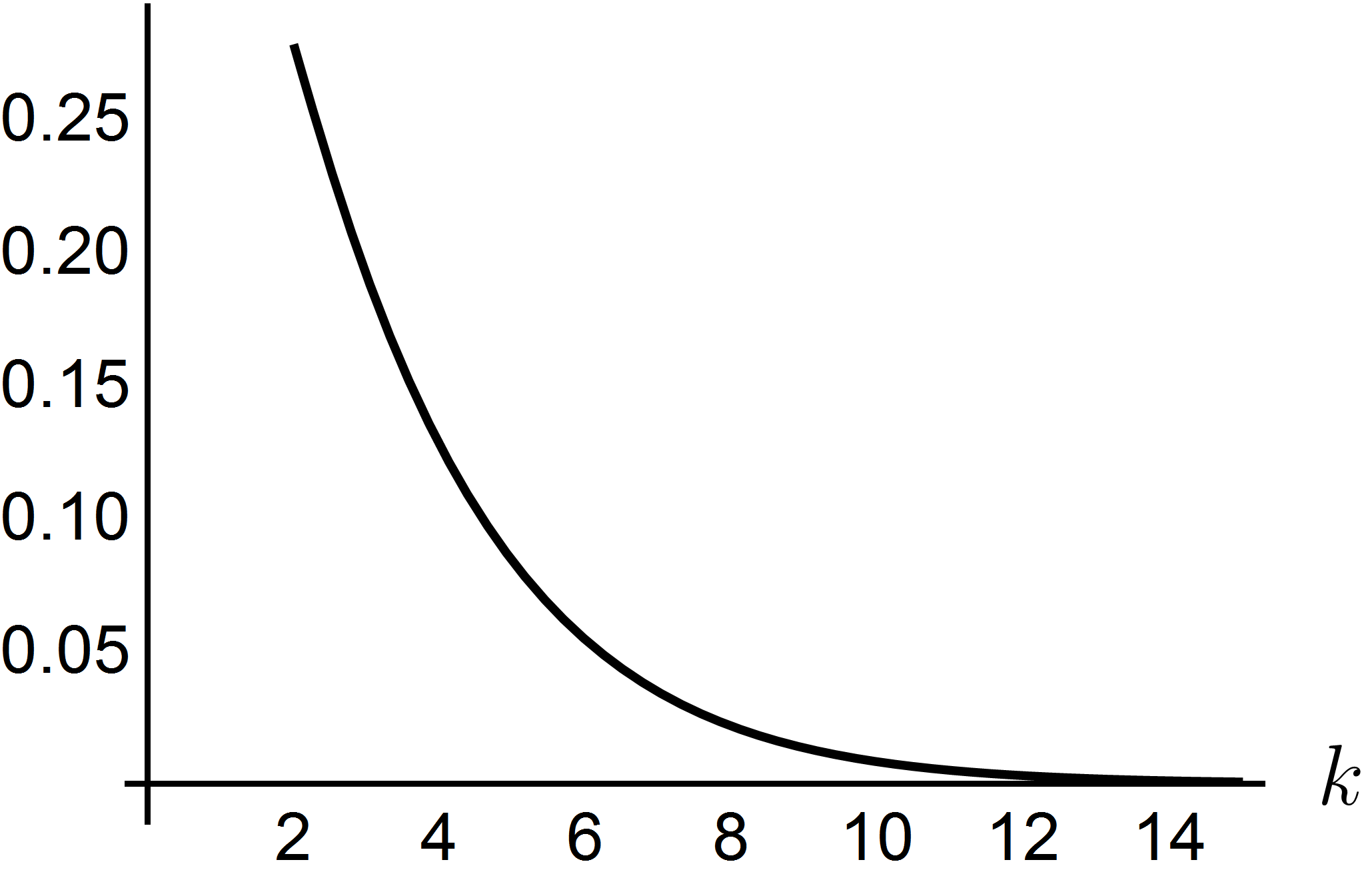} 
        \caption{\label{fig:F(k)}${F(k,\frac 15) - (k\cdot h(\frac 15)-1)}$.  (See Proposition \ref{prop:FBounds}).}
      
    \end{minipage}
\end{figure}

\begin{proposition} \label{prop:FBounds}
For an integer $k>1$ and $0<\gamma<\frac 12$, it holds that
$$kh(\gamma)-1\le F(k,\gamma) \le kh(\gamma)+\log(1+(1-2\gamma)^k)-1,$$
so,
$$
F(k,\gamma) = kh(\gamma) - 1 + O((1-2\gamma)^k)
$$ 
(see Figure \ref{fig:F(k)}).
\end{proposition}
\begin{proof}
The upper bound follows from Equation \ref{eqn:FByMinimization} which yields
$$F(k,\gamma) \le g\left(k,\gamma,\frac{\gamma}{1-\gamma}\right) = kh(\gamma)+\log(1+(1-2\gamma)^k)-1.$$

We turn to proving the lower bound. Clearly, 
$$
g(k,\gamma,x) \ge \log((1+x)^k)-\gamma k\log (x) - 1.
$$
The r.h.s.\ expression attains its minimum at $x = \frac \gamma{1-\gamma}$ and this minimum equals $kh(\gamma)-1$. Equation \ref{eqn:FByMinimization} implies that this is a lower bound on $F(k,\gamma)$.
\end{proof}

\subsection{Proof of Lemma \ref{lem:TProbability}}

We turn to the proof Lemma \ref{lem:TProbability}. It will be useful to view a vector $u\sim P$ as being generated in steps, with its $i$-th coordinate $u_i$ determined in the $i$-th step. The following proposition describes the quantities involved in this process.

\begin{proposition} \label{prop:p0p1Facts}
For $k\ge 2$ and $0 <\gamma < \frac 12$, let $u\in \bits{k}$ be a random vector sampled from $P$. For $0\le i\le k$, let $w_i$ denote the weight of the prefix vector $(u_1 ,\ldots, u_i)$. Then:
\begin{enumerate}
\item The distribution of the bit $u_i$ conditioned on the prefix $(u_1 ,\ldots, u_{i-1})$ depends only on the parity of $w_{i-1}$.

\item
\begin{equation}  \label{eqn:p_0to1}
\Pr(u_i=1\mid w_{i-1}\text{ is even}) =  \alpha\cdot\frac{(1+\alpha)^{k-i}-(1-\alpha)^{k-i}}{(1+\alpha)^{k-i+1}+(1-\alpha)^{k-i+1}} 
 \end{equation}
and
\begin{equation} \label{eqn:p_1to0}
\Pr(u_i=1\mid w_{i-1}\text{ is odd}) = 
\alpha\cdot\frac{(1+\alpha)^{k-i}+(1-\alpha)^{k-i}}{(1+\alpha)^{k-i+1}-(1-\alpha)^{k-i+1}}. 
\end{equation}
\end{enumerate}
\end{proposition}
\begin{proof}
Fix a prefix $(u_1,\ldots, u_{i-1})$ of weight $w_{i-1}$. We sum over $x = \|u\|-w_{i}$ and $y = \|u\|-w_{i-1}$.
\begin{align*}
\Pr\left(u_i=1\mid u_1,\ldots,u_{i-1}\right) &= \frac{\Pr\left(u_i=1\cap u_1,\ldots,u_{i-1}\right)}{\Pr\left( u_1,\ldots,u_{i-1}\right)} 
= \frac{\sum_{x\nequiv w_{i-1} \bmod 2} \binom{k-i}{x}\frac{\alpha^{x+w_{i-1}+1}}{Z}}{\sum_{y\equiv w_{i-1}\bmod 2} \binom{k-i+1}{y}\frac{\alpha^{y+w_{i-1}}}{Z}}\\
&= \alpha \frac{{\sum_{x\nequiv w_{i-1} \bmod 2}} \binom{k-i}{w-1}\alpha^{w-1}}{{\sum_{y\equiv w_{i-1}\bmod 2}} \binom{k-i+1}{w}\alpha^w},
\end{align*}
yielding the claim.
\end{proof}

We denote the r.h.s.\ of Equations \ref{eqn:p_0to1} and $\ref{eqn:p_1to0}$ by $p_{0\to 1,i}=p_{0\to1,i,k}$ and $p_{1\to 0,i}=p_{1\to 0,i,k}$, respectively. Also, for $0\le i\le k$, let 
$$e_i = e_{i,k}= \Pr_{u\sim P} (w_i\text{ is odd}).$$
Here are some useful facts about these terms.
Equation \ref{eqn:pExpectation} yields 
\begin{align*} \label{eqn:pegammaConvexCombination}
\gamma &= \Pr_{u\sim P}(u_i = 1) =p_{0\to 1,i}\cdot \Pr_{u\sim P} (w_{i-1}\text{ is even}) + p_{1\to 0,i}\cdot \Pr_{u\sim P} (w_{i-1}\text{ is odd})\\
&=  p_{1\to 0,i}e_{i-1} + p_{0\to 1,i}(1-e_{i-1}). \numberthis
\end{align*}
By similar considerations, we have
$$e_i = e_{i-1} \cdot (1-p_{1\to 0,i}) +(1- e_{i-1})\cdot p_{0\to 1,i}.$$
By combining these equations we find
\begin{equation}\label{eqn:p0to1Rel}
p_{0\to 1, i}\cdot (1-e_{i-1}) = \frac{\gamma + (e_i-e_{i-1})}2
\end{equation}
and
\begin{equation}\label{eqn:p1to0Rel}
p_{1\to 0, i}\cdot e_{i-1} = \frac{\gamma - (e_i - e_{i-1})}2.
\end{equation}

We need some further technical propositions.
\begin{proposition} \label{prop:pAndeBounds}
For every $\gamma \in (0,\frac 12)$ there exists some $c = c(\gamma) > 0$ such that if $k \ge 3$ then
$$e_{i,k}~,~~p_{0\to 1,i,k}~,~~p_{1\to 0,i,k} \in [c,1-c]$$
for every $1\le i\le k-1$.
\end{proposition}
\begin{proof}
It is not hard to see that both $p_{0\to 1,i,k}$ and $p_{1\to 0,i,k}$ are monotone in $i$. Therefore it suffices to check what happens for $i=1$ and for $i=k-1$. For $i=k-1$ the two terms equal $\frac{\alpha^2}{1+\alpha^2}$ and $\frac 12$ respectively. Since $\alpha$ is bounded from $0$ by Proposition \ref{prop:alphaAsymptotics}, this yields the claim.\\
For $i=1$ we note that $p_{0\to 1,1,k}=\gamma$.\\
It remains to consider $p_{1\to 0,1,k}$. Denote $x = \frac{1-\alpha}{1+\alpha}$ and note that $x$ is bounded away from $1$. This yields the bounds:
$$p_{1\to 0,1,k} = \frac{\alpha}{1+\alpha}\cdot\frac{1+x^{k-1}}{1-x^k} \ge \frac{\alpha}{1+\alpha}\cdot\frac{1-x}{1+x}$$
and
$$1-p_{1\to 0,1,k} = \frac{1}{1+\alpha}\cdot\frac{1-x^{k-1}}{1-x^k} \ge \frac{1}{1+\alpha}\cdot\frac{1-x}{1+x}$$

We turn to deal with $e_{i,k}$. Denote $a=1+\alpha$, $b=1-\alpha$ and $r=k-i-1$. A bound on $e_i$ follows from Equations \ref{eqn:pegammaConvexCombination} and \ref{eqn:gamma(alpha)} since
\begin{align*}
e_{i} &= \frac{\gamma - p_{0\to t,i+1}}{p_{1\to 0,i+1} - p_{0\to 1,i+1}} = \frac{\frac{a^{k-1}-b^{k-1}}{a^k+b^k}-\frac{a^{r-1}-b^{r-1}}{a^{r}+b^{r}}}{\frac{a^{r-1}+b^{r-1}}{a^{r}-b^{r}}-\frac{a^{r-1}-b^{r-1}}{a^{r}+b^{r}}} = \frac{(a^r-b^r)(a^{k-r}-b^{k-r})}{2(a^k+b^k)} \\
&=\frac{(1-x^r)(1-x^{k-r})}{2(1+x^k)}\ge \frac{(1-x)^2}{2(1+x)}
\end{align*}
and likewise,
$$1-e_i = \frac{(1+x^r)(1+x^{k-r})}{2(1+x^k)} \ge \frac{(1-x)^2}{2(1+x)}.$$
\end{proof} 

The following simple and technical proposition will come in handy in several situations below. It speaks about an experiment where $n$ balls fall randomly into $r$ bins. An {\em outcome} of such an experiment is an $r$-tuple of nonnegative integers $a_1,\ldots,a_r$ with $\sum a_i=n$, where $a_i$ is the number of balls at bin $i$ at the end of the experiment.

\begin{proposition} \label{prop:multinomialMode}
Let $r\ge 2$ be an integer $\frac{1}{r}\ge c > 0$, and $p_1,\ldots,p_r\ge c$ with $\sum p_i=1$. We drop randomly and independently $n$ balls into $r$ bins with probability $p_i$ of falling into bin $i$. The probability of every possible outcome is at most $O\left(n^{-\frac{r-1}2}\right)$, where $c, r$ are fixed and $n$ grows.
\end{proposition}
\begin{proof}
It is well known (e.g., \cite{Fel68} p.\ 171) that the most likely outcome of the above process $(a_1,\ldots,a_r)$, satisfies $np_i-1 < a_i$ for every $i$ and its probability is
\begin{align*}
\binom{n}{a_1,\ldots,a_r}\prod_{i=1}^r p_i^{a_i} &\le \binom{n}{a_1,\ldots,a_r}\prod_{i=1}^r \left(\frac{a_i+1}n\right)^{a_i} 
=\binom{n}{a_1,\ldots,a_r} \prod_{i=1}^r \left(\frac{a_i}n\right)^{a_i}\cdot\left(1+\frac{1}{a_i}\right)^{a_i}\\&
\le e^r\cdot \binom{n}{a_1,\ldots,a_r} \prod_{i=1}^r \left(\frac{a_i}n\right)^{a_i}
\le O\left(\frac{\sqrt n}{\prod_{i=1}^r \sqrt{a_i}}\right)\\&
\le O\left(\frac{\sqrt n}{\prod_{i=1}^r \sqrt{np_i-1}}\right) \le O\left(\frac{\sqrt n}{\sqrt{(cn-1)^r}}\right)\le O\left(n^{-\frac{r-1}2}\right)
\end{align*}
\end{proof}

\begin{proposition} \label{prop:BinomialSmallDeviation}
Let $a,c > 0$ be real and $n\in \mathbb N$ . Consider a random variable $X \sim B(n,p)$ where $c\le p\le 1-c$ and $np$ is an integer. Let $y$ be an integer such that $|y-pn| \le a \sqrt n$. Then $\Pr(X=y) \ge \Omega\left(n^{-\frac 12}\right)$ for fixed $a, c$ and $n\to\infty$.
\end{proposition}
\begin{proof}
Let $q = 1-p$, and let us denote $y = pn+x\sqrt n$, where $|x|\le a$.
\begin{align*}
\Pr(X=y) &= \binom ny p^y q^{n-y} = \binom ny \left(\frac yn\right)^{y}\left(\frac {n-y}n\right)^{n-y}\left(1-\frac {x\sqrt n}{y}\right)^y \left(1+\frac {x\sqrt n}{n-y}\right)^{n-y}
\end{align*}

Expand into Taylor Series, using the fact that $|x|$ is bounded and $y=\Theta(n)$ to derive the following inequalities:
$$\left(1-\frac{x\sqrt n}{y}\right)^y \ge \Omega\left(e^{-x\sqrt n}\right)~~~\text{and}~~~~\left(1+\frac{x\sqrt n}{n-y}\right)^{n-y} \ge \Omega\left(e^{x\sqrt n}\right).$$
The proposition now follows from Stirling's approximation, as
$$\binom ny \left(\frac yn\right)^{y}\left(\frac {n-y}n\right)^{n-y} \ge \Omega(n^{-\frac 12}).$$
\end{proof}

We are now ready to prove the main lemma of this section.

\lemTProbability*
\begin{proof}
Every binary ${k\times n}$ matrix $A$ that is sampled from the distribution $\pi$ satisfies the column condition, and we estimate the probability that the row condition holds.

By Proposition \ref{prop:pAndeBounds}, there is some $c=c(\gamma)>0$ so that $p_{0\to 1,i}~,~ p_{1\to 0,i}~,~ e_{i}$ are in $[c,1-c]$ for every $1\le i\le k-1$.

We recall that $A$'s columns are sampled independently and view $A$ as being sampled row by row. Let $b^i$ be the vector $A_{1}+\ldots+ A_{i-1} \bmod 2$. We want to observe how the ordered pairs $(\|b^{i}\|, \|A_{i}\|)$ evolve as $i$ goes from $1$ to $k$. By Proposition \ref{prop:p0p1Facts}, this evolution depends probabilistically on $\|b^{i-1}\|$ and only on it. Namely, let $s_i$ be the number of coordinates $j$ where $b^{i-1}_j=0$ and $A_{i,j}=1$. Likewise $t_i$ counts the coordinates $j$ for which $b^{i-1}_j=A_{i,j}=1$. It follows that $\|A_i\| = s_i+t_i$, and $\|b^{i}\|=\|b^{i-1}\|+s_i-t_i$, where $s_i\sim B(n-\|b^{i-1}\|, p_{0\to 1,i})$ and $t_i\sim B(\|b^{i-1}\|, p_{1\to 0,i})$ are independent binomial random variables.

Clearly $A\in T$ iff $\bigwedge_{i=1}^k D_i$, where $D_i$ is the event that $\|A_i\| = \gamma n$. 

We seek next an upper bound on $\Pr(A\in T)$.
\begin{align*}
\Pr(A\in T) &= \Pr(\bigwedge_{i=1}^k D_i) = \prod_{i=1}^k \Pr(D_i\;\mid\; \bigwedge_{j=1}^{i-1} D_j)\\
&\le \left(\prod_{i=1}^{k-3} \max_{w}\Pr(D_i \;\mid\; \|b_{i-1}\|=w)\right) \cdot \max_{w} \Pr(D_{k-2}\wedge D_{k-1}\wedge D_k \; \mid\; \|b_{k-3}\| = w).
\end{align*}
The inequality follows, since conditioned on $\|b_{i-1}\|$, the event $D_i$ is independent of $D_1,\ldots , D_{i-1}$.
We proceed to bound these terms. For $1\le i\le k-3$, 
$$\Pr(D_i \;\mid\; \|b_{i-1}\|=w) = \Pr(s_i+t_i = \gamma n \;\mid\; \|b_{i-1}\|=w).$$
If $w \ge \frac n2$, we condition on $s_i$ and bound this expression from above by 
$$\max_x \Pr(t_i = \gamma n - x \;\mid\; \|b_{i-1}\|=w \wedge s_i =  x),$$
namely, the probability that a $B(w, p_{1\to 0,i})$ variable takes a certain value. By Proposition \ref{prop:multinomialMode}, this is at most $O(w^{-\frac 12}) \le O(n^{-\frac 12})$. When $w < \frac n2$ the same argument applies with reversed roles for $t_i$ and $s_i$.

The last three rows of $A$ require a separate treatment, since e.g., the last row is completely determined by the first $k-1$ rows. Let $G$ be the matrix comprised of $A$'s last three rows. Denote $\epsilon:=b^{k-3}$, and let $w:=\|\epsilon\|$. Again it suffices to consider the case $w \ge \frac n2$, and similarly handle the complementary situation.
If $\epsilon_j=1$, the $j$-th column in $G$ must be one of the vectors
$(1,0,0)^\intercal, (0,1,0)^\intercal, (0,0,1)^\intercal, (1,1,1)^\intercal$. Let $a_1,a_2,a_3,a_4$ denote the number of occurrences of each of these vectors respectively. There are $n-w$ indices $j$ with $\epsilon_j=0$, and a corresponding column of $G$ must be one of the four even-weight vectors of length $3$. We condition on the entries of these columns. Under this conditioning $a_i+a_4$ is determined by the row condition applied to row $k-3+i$, and clearly also $\sum_1^4 a_i=w$. This system of four linearly independent linear equations has at most one solution in nonnegative integers. To estimate how likely it is that this unique solution is reached, we view it as a $w$-balls and $4$-bins experiment. The probability of each bin is a product of two terms from among $p_{0\to 1,i}~, 1-p_{0\to 1,i}~, p_{1\to 0,i}~, 1-p_{1\to 0,i}$ where $i\in \{k-2,k-1\}$. Again, these probabilities are bounded away from $0$.
By Proposition \ref{prop:multinomialMode} the probability of success is at most  $O(n^{-\frac 32})$. Consequently, $\Pr(A\in T) \le n^{-\frac k2}\cdot 2^{O(k)}$. 

To prove a lower bound on $\Pr(A\in T)$, again we consider the rows one at a time. As before, it is easier to bound the probability of $D_i$ by first conditioning on $\|b^{i-1}\|$. However, at present more care is needed, since letting the $\|b^{i}\|$'s take arbitrary values is too crude. Firstly, as long as the row conditions hold, necessarily $\|b^i\|$ is even. In addition, we monitor the deviation of $\|b^i\|$ from its expectation, which is $n\cdot e_i$. Accordingly, we define the following sets:
$$\text{For~~}1\le i\le k-2,~~\text{let~~~} S_i := \{0 \le w\le n \mid ~~|w-e_i\cdot n| \le \sqrt n ~~\wedge ~~w \text{ is even}\}.$$
The intuition is that the event $\|b^i\|\in S_i$ makes it likely that $D_{i+1}$ holds, in which case it is also likely that $\|b^{i+1}\|\in S_{i+1}$. This chain of probabilistic implication yields our claim. To start, clearly $\|b^0\| \in S_0 := \{0\}$. 

Now,
\begin{align*}&\Pr(A\in T) = \Pr\left(\bigwedge_{i=1}^k D_i\right) \ge \Pr\left(\bigwedge_{i=1}^k D_i \wedge \bigwedge_{i=1}^{k-2} \|b^i\|\in S_i\right) \\
&= \left(\prod_{i=1}^{k-2} \Pr\left((D_i\wedge \|b^i\|\in S_i) \;\mid\; \bigwedge_{j=1}^{i-1} (D_j\wedge \|b^j\|\in S_j)\right)\right) \cdot \Pr\left((D_{k-1}\wedge D_k) \; \mid \; \bigwedge_{j=1}^{k-2}(D_j\wedge \|b^j\|\in S_j)\right) \\
&\ge \left(\prod_{i=1}^{k-2}\min_{w\in S_{i-1}} \Pr((D_i\wedge \|b^i\|\in S_i) \;\mid\; \|b^{i-1}\| = w) \right) \cdot \min_{w\in S_{k-2}} \Pr((D_{k-1}\wedge D_k) \; \mid \; \|b^{k-2}\| = w).
\end{align*}
It is in estimating these last terms that the assumption $\|b^i\|\in S_i$ becomes useful. We proceed to bound these terms, and claim the following:
\begin{enumerate}
\item $\min_{w\in S_{i-1}} \Pr((D_i\wedge \|b^i\|\in S_i) \;\mid\; \|b^{i-1}\| = w) \ge \Omega(\frac 1{\sqrt n})$ for every $1\le i\le k-2$.
\item $\min_{w\in S_{k-2}} \Pr((D_{k-1}\cap D_k) \; \mid \; \|b^{k-2}\| = w) \ge  \Omega(\frac{1}{n})$. 
\end{enumerate}
It is clear that the above inequalities imply that $\Pr(A\in T) \ge {n^{\frac k2}}\cdot 2^{-O(k)}$, which proves the lemma.

Fix some $1 \le i \le k-2$ and let $w\in S_{i-1}$, and assume that $D_i$ holds. Then
$$\|b^i\|-\|b^{i-1}\| \equiv s_i - t_i \equiv s_i + t_i \equiv \gamma n\equiv 0 \mod 2,$$
so that $\|b^i\|$ satisfies $S_i$'s parity condition. Therefore
$$ 
\Pr(D_i\wedge \|b^i\|\in S_i \;\mid\; \|b^{i-1}\| = w) = \Pr(D_i\wedge |\|b^i\|-\E(\|b^i\|)|\le \sqrt n \;\mid\; \|b^{i-1}\| = w)
$$
Namely
\begin{align*} \label{eqn:Pr(AinT)lowerBound1}
&\Pr(D_i\wedge \|b^i\|\in S_i \;\mid\; \|b^{i-1}\| = w) \\ &=\Pr(s_i+t_i = \gamma n \wedge |s_i - t_i - e_i\cdot n + w| \le \sqrt n\;\mid\; \|b^{i-1}\| = w). \numberthis
\end{align*}
We want to express this last condition in terms of $x = s_i - t_i$, where clearly $s_i = \frac{\gamma n+x}{2}$ and $t_i = \frac{\gamma n-x}{2}$. Equation \ref{eqn:Pr(AinT)lowerBound1} means that $e_i\cdot n-w- \sqrt n \le x\le e_i\cdot n-w + \sqrt n$ and $x\equiv \gamma n\mod 2$. Summing over all such $x$'s we have
\begin{equation} \label{eqn:Pr(AinT)lowerBound2}
\Pr(D_i\wedge \|b^i\|\in S_i \;\mid\; \|b^{i-1}\| = w) = \sum_x \Pr(s_i = \frac{\gamma n+x}{2}) \cdot \Pr(t_i = \frac{\gamma n-x}{2}).
\end{equation}
Here $s_i \sim B(n-w,p_{0\to1,i})$ and $t_i\sim B(w,p_{1\to 0,i})$. We use Proposition \ref{prop:BinomialSmallDeviation} to give lower bounds on a general term in Equation \ref{eqn:Pr(AinT)lowerBound2}. To this end we show that $\frac{\gamma n+x}2$ and $\frac{\gamma n -x}{2}$ are close, respectively, to the means of $s_i$ and $t_i$. 

Since $w\in S_{i-1}$, we can write $w=e_{i-1}\cdot n + y$ where $|y| \le \sqrt n$. The bounds on $x$ allow us to write $x = (e_i - e_{i-1})n - y + z$ for some $|z| \le \sqrt n$. By Equation \ref{eqn:p0to1Rel},
\begin{align*}
\left|\E(s_i)-\frac{\gamma n+x}2\right| &= \left|p_{0\to 1,i}\cdot (n-w)-\frac{\gamma n+x}2\right| \\ &= \left|p_{0\to 1,i}\cdot ((1-e_{i-1})n - y)-\frac{(\gamma+e_i-e_{i-1})n-y+z}2\right| 
\\&= \left|\frac{\gamma+(e_i-e_{i-1})}{2}n-p_{0\to 1,i}\cdot y-\frac{(\gamma+e_i-e_{i-1})n-y+z}2\right| \\
&=\left|\frac {y-z}2  - p_{0\to1,i}\cdot y \right| \le \sqrt n.
\end{align*}
By Proposition \ref{prop:BinomialSmallDeviation}, $\Pr(s_i = \frac{\gamma n+x}2) \ge \Omega(n^{-\frac 12})$. A similar proof, using Equation \ref{eqn:p1to0Rel}, shows that $\Pr(t_i = \frac{\gamma n-x}2 \ge \Omega(n^{-\frac 12}))$. Thus, each of the  $\Omega(\sqrt n)$, summands in Equation \ref{eqn:Pr(AinT)lowerBound2} is at least $\Omega(n^{-1})$, so that
$$\Pr(D_i\wedge \|b^i\|\in S_i \;\mid\; \|b^{i-1}\| = w) \ge \Omega(n^{-\frac 12}).$$

We turn to proving a lower bound on $\min_{w\in S_{k-2}} \Pr((D_{k-1}\wedge D_k) \; \mid \; \|b^{k-2}\| = w)$. The column condition implies that $A_k = b^{k-1}$. Thus, for  $w\in S_{k-2}$,
\begin{align*}
&\Pr((D_{k-1}\wedge D_k) \; \mid \; \|b^{k-2}\| = w) = \Pr(D_{k-1}\wedge \|b^{k-1}\| = \gamma n\mid \|b^{k-1}\| = w) \\
&=\Pr(s_{k-1}+t_{k-1}=\gamma n\wedge s_{k-1}-t_{k-1}+w=\gamma n) \\
&= \Pr\left(s_{k-1}=\gamma n-\frac w2\right)\cdot \Pr\left(t_{k-1}=\frac w2\right),
\end{align*}
where $s_{k-1}\sim B(n-w,p_{0\to 1,k-1})$ and $t_{k-1}\sim B(w,p_{1\to 0,k-1})$. Again, by applying Proposition \ref{prop:BinomialSmallDeviation} to $s_{k-1}$ and $t_{k-1}$, we conclude that the above is at least $\Omega(n^{-1})$. 
\end{proof}
      
\section{Bounding $|T_V|$ in general}\label{sec:MatrixEnumerationGeneral}
In this section we fix a robust subspace $V\le \bits{k}$ and bound its contribution to Equation \ref{eqn:k'thMomentAfterMoebius}. Let us sample, uniformly at random a matrix $A_{k\times n}$ in $T_V$. Since $T_V$ is invariant under column permutations, the columns of $A$ are equally distributed. We denote this distribution on $\bits{k}$ by $Q_V$, and note that 
$$\log |T_V| = h(A) \le n\cdot h(Q_V).$$

To bound $h(Q_V)$ we employ the following strategy. Express $V$ as the kernel of a $(k-d(V)) \times k$ binary matrix $B$ in reduced row echelon form. Suppose that $B_{i,j} = 1$. If $B_{i',j} = 0$ for every $i'<i$ we say that the coordinate $j$ is $i$-{\em new}. Otherwise, $j$ is said to be $i$-{\em old}. We denote the set of $i$-new coordinates by $\Delta_i$. We have assumed that $V$ is robust, so that $\bigcup_{i=1}^{k-d} \Delta_i = [k]$, since $j\not\in\bigcup_{i=1}^{k-d} \Delta_i$ means that coordinate $j$ is sensitive. Also $B$ is in reduced row echelon form, so all $\Delta_i$ are nonempty.

\begin{example}
The following $B_{3\times 7}$ corresponds to $k=7$ and $d(V) = 4$. In bold - the $i$-new entries in row $i$ for $i=1,2,3$.
$$
\begin{bmatrix}
\bf1 & 0 & 0 & \bf1 & \bf1 & 0 & 0 \\ 
0 & \bf1 & 0 & 1 & 0 & \bf1 & \bf1 \\ 
0 & 0 & \bf1 & 1 & 0 & 1 & 0 \\ 
\end{bmatrix}
$$
\end{example}

A vector $v$ sampled from $Q_V$ satisfies $Bv = 0$ and the expected value of each of its coordinates is $\E(v_i) = \gamma$. Consider $v$ as generated in stages, with the coordinates in $\Delta_i$ determined in the $i$-th stage. We express $v$'s entropy in this view:
\begin{equation} \label{eqn:h(v)InStages}
h(Q_V) = h(v) = h(v_{\Delta_1}) + \sum_{i=2}^{k-d(V)} h(v_{\Delta_i}\mid v_{\bigcup_{i'=1}^{i-1}\Delta_{i'}}).
\end{equation}

We begin with the first term. Since $\Delta_1$ is the support of $B$'s first row and since $Bv=0$, it follows that $v_{\Delta_1}$ has even weight. As we show in Lemma \ref{lem:PHasMaximumEntropy}, the distribution $P$ from Section \ref{sec:MatrixEnumerationEvens} has the largest possible entropy for a distribution that is supported on even weight vectors with expectation $\gamma$ per coordinate. Hence, 
$$h(v_{\Delta_1}) \le h(P_{{|\Delta_1|,\gamma}}) = F(|\Delta_1|,\gamma)$$

It takes more work to bound the other terms in Equation \ref{eqn:h(v)InStages}. Let $2\le i\le k-d(V)$. 
Before the $i$-th stage, $v$'s $i$-old coordinates are already determined. Since the inner product $\langle B_i,v\rangle=0$, the $i$-new coordinates of $v$ have the same parity as its $i$-old coordinates. Hence $\|v_{\Delta_i}\|$'s parity is determined before this stage. Let $\delta_i = \Pr(\|v_{\Delta_i}\|\text{ is odd})$. Since conditioning reduces entropy 
$$h(v_{\Delta_i}\mid v_{\bigcup_{i'=1}^{i-1}\Delta_{i'}}) \le h(v_{\Delta_i}\mid \text{parity of } \|v_{\Delta_i}\|) = h(v_{\Delta_i}) - h(\delta_i).$$

We have already mentioned that Lemma \ref{lem:PHasMaximumEntropy} characterizes the max-entropy distribution on even-weight vectors with given per-coordinate expectation. We actually do more, and find a  maximum entropy distribution $P=P_{m,\gamma,\delta}$ on $\bits{m}$ satisfying
\begin{equation} \label{eqn:PRowCondition}
\Pr_{u\sim P}(u_i=1)=\gamma
\end{equation}
for every $1\le i\le m$ and
\begin{equation} \label{eqn:PColumnCondition}
\Pr_{u\sim P}(\|u\|\text{ is odd})=\delta.
\end{equation} 

This distribution $P=P_{m,\gamma,\delta}$ extends something we did before, in that $P_{m,\gamma,0}$ coincides with $P_{m,\gamma}$ from Section \ref{sec:MatrixEnumerationEvens}. 

Since $v_{ֿ|\Delta_i|}$ also satisfies these conditions, this yields the bound $h(v_{\Delta_i}\mid v_{\bigcup_{i'=1}^{i-1}\Delta_{i'}}) \le F(|\Delta_i|, \gamma, \delta_i)$, where $F(m,\gamma,\delta) = h(P_{m,\gamma,\delta})-h(\delta)$. We conclude that
\begin{equation} \label{eqn:T_VBound}
\log|T_V|\le n\cdot h(Q_V) \le n\cdot \left(F(|\Delta_1|,\gamma)+\sum_{i=2}^{k-d(V)}F(|\Delta_i|,\gamma,\delta_i)\right).
\end{equation}

The relevant consistency relation is that $F(m,\gamma,0) = F(m,\gamma)$.
We determine next the distribution $P_{m,\gamma,\delta}$ and then return to the analysis of Equation \ref{eqn:T_VBound}.

\subsection{The function $F(m,\gamma,\delta)$} 
As explained above we now find the max-entropy distribution satisfying Equations \ref{eqn:PRowCondition} and \ref{eqn:PColumnCondition}. The following proposition gives a necessary condition for the existence of such a distribution.

\begin{proposition}\label{prop:p_condition}
If there is a distribution satisfying conditions \ref{eqn:PRowCondition} and \ref{eqn:PColumnCondition}, then $\gamma \ge \gammin$, where $\gammin = \frac \delta m$.
\end{proposition}
\begin{proof}
Let $P$ be such a distribution and let $u\sim P$. By Equation \ref{eqn:PRowCondition}, $\E(\|u\|) = \gamma m$. The lower bound on $\gamma$ follows since each odd vector weighs at least $1$ and thus
$$\delta = \Pr(\|u\|\text{ is odd}) \le \E(\|u\|).$$
\end{proof}
\begin{remark}
As we show soon, the condition in Proposition \ref{prop:p_condition} is also sufficient. 
\end{remark}

Let $m\ge 2$ and assume that $m,\gamma,\delta$ satisfy the strict inequalities $0<\delta< 1$ and $\gammin < \gamma$. We define the distribution $P=P_{m,\gamma,\delta}$ on $\bits{m}$ as follows:

\begin{equation} \label{eqn:PGeneralDefinition}
P(u) = \begin{cases}
\frac {\alpha^{\|u\|}}Z & \text{if }\|u\|\text{ is even}\\
\frac {\beta\cdot\alpha^{\|u\|}}Z & \text{if }\|u\|\text{ is odd}\\
\end{cases}
\end{equation}
where
$$Z = \sum_{u\in \evens{m}} \alpha^{\|u\|}+\beta\sum_{u\in \odds{m}} \alpha^{\|u\|} = \frac{(1+\beta)(1+\alpha)^m+(1-\beta)(1-\alpha)^m}{2}.$$
As we show there exist unique positive reals $\alpha$, $\beta$ for which Equations \ref{eqn:PRowCondition} and \ref{eqn:PColumnCondition} hold. Note that 
$$\Pr_{u\sim P}(\|u\|\text{ is odd}) = \frac{\beta\left((1+\alpha)^m-(1-\alpha)^m\right)}{2Z},$$
so Equation \ref{eqn:PColumnCondition} is equivalent to
$$
\beta = \frac{\delta}{1-\delta}\cdot\frac{(1+\alpha)^m+(1-\alpha)^m}{(1+\alpha)^m-(1-\alpha)^m},$$
showing in particular that $\alpha$ determines the value of $\beta$. Substituting the above into Equation \ref{eqn:PRowCondition} gives
\begin{align*} 
\gamma &= \Pr(u_i = 1) = \alpha\frac{(1+\beta)(1+\alpha)^{m-1}+(1-\beta)(1-\alpha)^{m-1}}{2Z} \\
&= \alpha(1-\delta)\frac{(1+\alpha)^{m-1}-(1-\alpha)^{m-1}}{(1+\alpha)^m+(1-\alpha)^m}+\alpha\delta\frac{(1+\alpha)^{m-1}+(1-\alpha)^{m-1}}{(1+\alpha)^m-(1-\alpha)^m}.
\end{align*}
Denote the right side of this expression by $\gamma(m,\alpha,\delta)$. The following generalizes Proposition \ref{prop:gammaIncreases}.

\begin{proposition}
Let $m\ge 2$. In the range $1 > \alpha > 0$ the function $\gamma(m,\alpha,\delta)$ increases from $\gammin$ to $\frac 12$. 
\end{proposition}
\begin{proof}
Clearly, it is enough to prove the proposition for $\delta=0,1$. The case $\delta=0$ was dealt with in Proposition \ref{prop:gammaIncreases}. The same argument works for $\delta = 1$ as well, since
$$\gamma = \alpha\frac{(1+\alpha)^{m-1}+(1-\alpha)^{m-1}}{(1+\alpha)^m-(1-\alpha)^m} = \frac{\sum_{i\text{ odd}}\binom{m-1}{i-1}\alpha^i}{\sum_{i\text{ odd}}\binom{m}{i}\alpha^i}.$$
\end{proof}

Hence, $\gamma(m,\alpha,\delta)$ has an inverse with respect to $\alpha$, which we denote $\alpha(m,\gamma,\delta)$. The uniqueness of $\alpha$ and $\beta$ follows.

We can also define $P$ at the extreme values $\delta \in \{0,1\}$ and $\gamma=\gammin$ by taking limits in Equation \ref{eqn:PGeneralDefinition}. The limit $\alpha\to 0$ corresponds to $\gamma=\gammin$ and $\beta\to 0$ resp. $\beta\to\infty$ to $\delta=0$ or $\delta=\infty$. We still require, however, that $\gamma > 0$. E.g., if $\gamma=\gammin$, $P$ yields each weight $1$ vector with probability $\frac \delta m$ and the weight $0$ vector with probability $1-\delta$. Also, as already mentioned $P_{m,\gamma,0}$ coincides with $P_{m,\gamma}$ from Section \ref{sec:MatrixEnumerationEvens}.

We turn to compute $P$'s entropy:
\begin{align*} \label{eqn:PGeneralEntropy}
h(P) &= -\sum_{u\in \evens{m}}\frac{\alpha^{\|u\|}}{Z}\log\frac{\alpha^{\|u\|}}{Z}-\sum_{u\in \odds{m}}\frac{\beta\alpha^{\|u\|}}{Z}\log\frac{\beta\alpha^{\|u\|}}{Z}
\\ &= \log Z - \delta\log\beta  - \gamma m \log\alpha \\ &= h(\delta)+(1-\delta)\log((1+\alpha)^m+(1-\alpha)^m)+\delta\log((1+\alpha)^m-(1-\alpha)^m) \\ &~~~-\gamma m\log\alpha - 1 \numberthis
\end{align*}
and recall that $F(m,\gamma,\delta) = h(P) - h(\delta)$. Consistency for the boundary cases $\delta\in\{0,1\}$ or $\gamma = \gammin$ follows by continuity and passage to the limit. In particular, $F(m,\gamma,0) = F(m,\gamma)$. Also, let $F(m,\gamma,\delta) = -\infty$ for $\gamma < \gammin$.

For $\gammin<\gamma < \frac 12$ we also have the following generalization of Equation \ref{eqn:FByMinimization}, which follows from the same argument:
\begin{equation} \label{eqn:FByMinimizationGeneral}
F(m,\gamma,\delta) = \min_{x>0} g(m,\gamma,x,\delta)
\end{equation}
where
$$
g(m,\gamma,x,\delta) = (1-\delta)\log\left((1+x)^m+(1-x)^m\right)+\delta\log\left((1+x)^m-(1-x)^m\right)-\gamma m\log x - 1$$
with the minimum attained at $x=\alpha$.

We are now ready to show that $P$ is the relevant max-entropy distribution.

\begin{lemma}  \label{lem:PHasMaximumEntropy}
Fix $m\ge 2$, $0\le \delta\le 1$ and $\gammin\le \gamma < \frac 12$. The largest possible entropy of a $\bits{m}$-distribution satisfying Equations \ref{eqn:PRowCondition} and \ref{eqn:PColumnCondition}, is $h(P_{m,\gamma,\delta})$. 
\end{lemma}
\begin{proof}
Let $\cal D$ denote the polytope of $\bits{m}$-distributions that satisfy Conditions \ref{eqn:PRowCondition} and \ref{eqn:PColumnCondition}. Note that if $\gamma = \gammin$ this polytope is reduced to a point, and the claim is trivial. We henceforth assume that $\gammin<\gamma$, and seek a distribution $Q\in \cal D$ of maximum entropy. This distribution is unique, since the entropy function is strictly concave. Also, the value of $Q(u)$ depends only on $\|u\|$ for all $u\in \bits{m}$, since the optimum is unique and this maximization problem is invariant to permutation of coordinates in $\bits{m}$. 

Let $a_i = Q(u)$ where $\|u\| = i$. We claim that 
\begin{equation} \label{eqn:PHasMaximumEntropyRelation1}
a_{i-2}\cdot a_{i+2} = a_{i}^2
\end{equation}
for every $2\le i\le m-2$. Indeed, let $x,y,y',z\in \bits{m}$ be the indicator vectors for, respectively, the sets $\{3,\ldots,i\}$,  $\{1,\ldots,i\}$, $\{3,\ldots,i+2\}$ and $\{1,\ldots,i+2\}$. Consider the distribution $Q+\theta$ where
$$\theta(u) = \begin{cases}
\epsilon& \text{for }u = y,y'\\
-\epsilon& \text{for }u = x,z\\
0 &\text{otherwise}.
\end{cases}
$$
Note that, if $a_{i-2},a_{i},a_{i+2}$ are positive, $Q+\theta\in \cal D$ for $|\epsilon|$ small enough. Hence, by the optimality of $Q$,
$$
0 = \nabla_\theta h(Q) = \log{\frac{a_{i-2}a_{i+2}}{a_{i}^2}},
$$
yielding Equation \ref{eqn:PHasMaximumEntropyRelation1}.

We also want to rule out the possibility that exactly one side of Equation \ref{eqn:PHasMaximumEntropyRelation1} vanishes. However, even if exactly one side vanishes, it is possible to increase $h(Q)$ by moving in the direction of either $\theta$ or $-\theta$.

A similar argument yields 
\begin{equation} \label{eqn:PHasMaximumEntropyRelation2}
a_i\cdot a_{i+3} = a_{i+1}\cdot a_{i+2}
\end{equation}
for $0\le i\le m-3$. Here, we take
$$
\theta(u) = \begin{cases}
\epsilon& \text{for }u = x,w\\
-\epsilon& \text{for }u = y,z\\
0 &\text{otherwise}.
\end{cases}
$$
where $x,y,z,w$ are the respective indicator vectors of $\{3,\ldots,i+2\}$, $\{3,\ldots,i+3\}$, $\{1,\ldots,i+2\}$ and $\{1,\ldots,i+3\}$.

Equation \ref{eqn:PHasMaximumEntropyRelation1} and \ref{eqn:PHasMaximumEntropyRelation2} imply that one of the following must hold:
\begin{enumerate}
\item $a_0,a_2,\ldots,a_{2\floor{\frac m2}}$ and $a_1,a_3,\ldots,a_{2\floor{\frac {m-1}2}+1}$ are geometric sequences with the same positive quotient.
\item $a_0 = (1-\delta)$, $a_1 = \delta$ and $a_i = 0$ for every $i\ge 2$.
\item $a_{m-1}$ and $a_m$ are $\delta$ and $1-\delta$ according to $m$'s parity, and $a_i=0$ for all $i\le m-2$. 
\end{enumerate}

Case 2 corresponds to $\gamma = \gammin$ and case 3 is impossible since $\gamma < \frac 12$, so we are left with case 1. If $0<\delta<1$, note that $Q$ must satisfy Equation \ref{eqn:PGeneralDefinition} for some positive $\alpha$ and $\beta$. By the uniqueness of these parameters, it follows that $Q=P$. 

If $\delta = 0,1$ then $a_i$ vanishes for odd resp.\ even $i$'s. Thus, $Q$ satisfies Equation \ref{eqn:PGeneralDefinition} with $\beta$ going to $0$ or $\infty$. 
\end{proof}

\subsection{Properties of $F(m,\gamma,\delta)$}

Our analysis of Equation \ref{eqn:T_VBound} requires that we understand $F$'s behavior in certain regimes.

\begin{figure}
    \centering
    \includegraphics[width=0.5\textwidth]{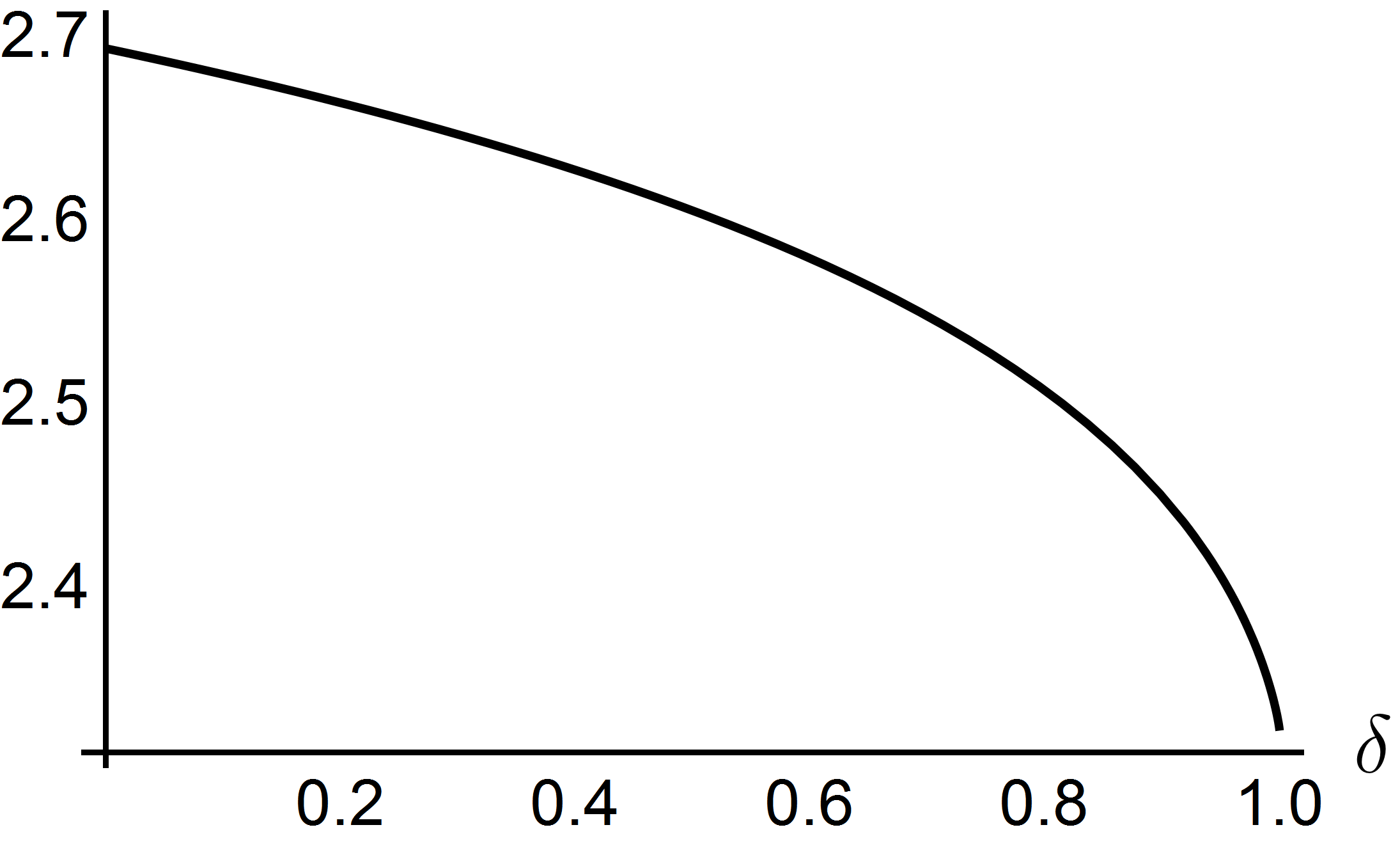} 
    \caption{Illustration for Lemma \ref{lem:FMonotoneInDelta} - \label{fig:FofDelta}$F(5,\frac 15,\delta)$}
\end{figure}

\begin{lemma} \label{lem:FMonotoneInDelta}
If $m>1$ is an integer, and $0<\gamma < \frac 12$, then $F(m,\gamma,\delta)$ is a non-increasing function of $\delta$ (see Figure \ref{fig:FofDelta}).
\end{lemma}
\begin{proof}
If $\delta > \gamma m$, then $\gamma < \gammin$ and $F(m,\gamma,\delta) = -\infty$. It suffices, therefore, to consider the range $0\le \delta < \gamma m$.

Let $0\le \delta < \delta' < \gamma m$ and let $\alpha = \alpha(m,\gamma,\delta)$. By Equations \ref{eqn:FByMinimizationGeneral} and \ref{eqn:PGeneralEntropy}:
\begin{align*}
&~~~~ F(m,\gamma,\delta') - F(m,\gamma,\delta) \le g(m,\alpha,\delta') - F(m,\gamma,\delta) \\ &
=(\delta'-\delta)\left(\log\left((1+\alpha)^m-(1-\alpha)^m\right)-\log\left((1+\alpha)^m+(1-\alpha)^m\right)\right) \le 0
\end{align*} 
\end{proof}

We now return to the case $\delta = 0$, and discuss the convexity of $F$ in this regime.

\begin{lemma} \label{lem:FConvexInk}
For any $0<\gamma < \frac 12$ the function $F(m,\gamma)$ is strictly convex in $m$ for $m \ge 2$. (See Figure \ref{fig:F(k)}).
\end{lemma}
\begin{proof}
Since $\gamma$ is fixed throughout the proof, we can and will denote $F(m) = F(m,\gamma)$, $g(m,x) = g(m,\gamma,x)$. Also, $\alpha = \alpha(m,\gamma)$ is the value of $x$ which minimizes $g(m,\gamma,x)$. This allows us to extend the definition of $\alpha$ to real $m$. Note that Equation \ref{eqn:gamma(alpha)} still holds in this extended setting, and that $1> \alpha> 0$. In addition, $a = 1+\alpha$ and $b = 1-\alpha$. 

Our goal is to show that for $m\ge2$ there holds
$$\frac{\partial^2 F}{\partial m^2}(m,\alpha) \ge 0.$$
It follows from  Equation \ref{eqn:FByMinimization} that 
\begin{equation} \label{eqn:dg/dxVanishes}
\frac{\partial g}{\partial x}(m,\alpha) = 0.
\end{equation}
Taking the derivative w.r.t.\ $m$ yields
\begin{equation} \label{eqn:dalpha/dm}
\frac{\partial^2 g}{\partial x \partial m}(m,\alpha) + \frac{\partial^2 g}{\partial x^2}(m,\alpha)\frac{d\alpha}{dm} = 0.
\end{equation}
Using Equation \ref{eqn:dg/dxVanishes} we obtain:
\begin{align*} 
\frac{\partial F}{\partial m} &= \frac{\partial g}{\partial m}(m,\alpha) + \frac{\partial g}{\partial x}(m,\alpha)\frac{d \alpha}{dm} =\frac{\partial g}{\partial m}(m,\alpha).
\end{align*}
Next,
$$
\frac{\partial^2 F}{\partial m^2} = \frac{\partial^2 g}{\partial m^2}(m,\alpha) + \frac{\partial^2 g}{\partial m\partial  x}(m,\alpha)\frac{d \alpha}{dm} = \frac{\partial^2 g}{\partial m^2}(m,\alpha) - \frac{\left(\frac{\partial^2 g}{\partial m \partial x}(m,\alpha)\right)^2}{\frac{\partial^2 g}{\partial x^2}(m,\alpha)}
$$
where the second equality follows from Equation \ref{eqn:dalpha/dm}. The partial derivatives commute since $g$ is smooth. We claim that $\frac{\partial^2 g}{\partial x^2} > 0$. To this end we refer to the definition of $g$ in Equation \ref{eqn:CrossEntropy} and take its derivative twice, then use the defining relation relation between $\gamma$ and $\alpha$ (Equation \ref{eqn:gamma(alpha)}) to see that the sign of this derivative is the same as that of
\begin{align*}
&~~~~(m-1)(a^{m-2}+b^{m-2})(a^m+b^m)-m(a^{m-1}-b^{m-1})^2+\frac{(a^m+b^m)(a^{m-1}-b^{m-1})}{\alpha}\\
&> (m-1)(a^{m-2}+b^{m-2})(a^m+b^m)-m(a^{m-1}-b^{m-1})^2+(a^m+b^m)(a^{m-1}-b^{m-1})\\
&> 0.
\end{align*}
Thus, to prove the lemma it suffices to show that
$$
\frac{\partial^2 g}{\partial m^2}(m,\alpha)\frac{\partial^2 g}{\partial x^2}(m,\alpha) > \left(\frac{\partial^2 g}{\partial m\partial x}(m,\alpha)\right)^2
$$
when $m\ge2$.

We wish to show that $rs > t^2$, where
\begin{align*}
r &= \ln2(a^m+b^m)^2\frac{\partial^2 g}{\partial m^2}(m,\alpha) \\
s &= \ln2(a^m+b^m)^2\frac{\partial^2 g}{\partial x^2}(m,\alpha) \\
t &= \ln2(a^m+b^m)^2\frac{\partial^2 g}{\partial m\partial x}(m,\alpha)
\end{align*}
We start with the first order derivatives
$$
\frac{\partial g}{\partial m}(m,\alpha) = \frac{a^m\log a+b^m\log b}{a^m+b^m}-\gamma\log x
$$
and
$$
\frac{\partial g}{\partial x}(m,\alpha) = \frac{m(a^{m-1}-b^{m-1})}{a^m+b^m}-\frac{m\gamma}x.
$$

Expand the second order derivatives with $\gamma$ replaced according to Equation \ref{eqn:gamma(alpha)} to get
\begin{align*}
r &= m(m-1)(a^{m-2}+b^{m-2})(a^m+b^m)-m^2(a^{m-1}-b^{m-1})^2+\frac{m\gamma(a^m+b^m)^2}{\alpha^2} \\
&=m\left((m-1)(a^{m-2}+b^{m-2})(a^m+b^m)+\frac{(a^{m-1}-b^{m-1})(a^m+b^m)}{\alpha}-m(a^{m-1}-b^{m-1})^2 \right) \\
&> m\left((m-1)(a^{m-2}+b^{m-2})(a^m+b^m)+(a^{m+2}+b^{m+2})(a^m+b^m)-m(a^{m-1}-b^{m-1})^2\right)\\
&= 4m^2a^{m-2}b^{m-2}.
\end{align*}
The inequality follows from ${a^{m-1}-b^{m-1}}-\alpha(a^{m-2}+b^{m-2}) = \frac{a+b}{2}(a^{m-2}-b^{m-2}) > 0$.
Also
\begin{align*}
s &= (a^m+b^m)(a^m(\log a)^2 + b^m (\log b)^2)-(a^m\log a+b^m\log b)^2 \\
&= a^mb^m(\log a - \log b)^2. 
\end{align*}
and
\begin{align*}
t &= \left((m\log a+1)a^{m-1}-(m\log b+1)b^{m-1}\right)(a^m+b^m)\\
&-m(a^{m-1}-b^{m-1})(a^m\log a+b^m\log b) - \frac{\gamma(a^m+b^m)^2}{\alpha}\\
&= 2ma^{m-1}b^{m-1}(\log a-\log b).
\end{align*}
We therefore conclude that
$$rs > t^2$$
as claimed.
\end{proof}

The following corollary follows immediately from Lemma \ref{lem:FConvexInk}.

\begin{corollary} \label{cor:FConvexDiscrete}
For every $0<\gamma <\frac 12$ and every $2\le m\le m'$, the holds
$$F(m',\gamma) + F(m,\gamma) < F(m'+1,\gamma) + F(m-1, \gamma).$$
\end{corollary}

We also need the following result in order to bound $|T_V|$.
\begin{proposition} \label{prop:F12Bound}
Let $0 < \gamma < \frac 12$, $0\le \delta\le 1$ and $m\ge 2$. Then,
$$F(1,\gamma,\delta)+F(m+1,\gamma) < F(2,\gamma)+F(m,\gamma).$$
\end{proposition}
\begin{proof}
Recall that $F(1,\gamma,\delta) \le 0$ and $F(2,\gamma)=h(\gamma)$. Thus, the claim follows from
$$F(m+1,\gamma) < F(m,\gamma) + h(\gamma).$$
This holds since $F$ is strictly convex in $m$ (Lemma \ref{lem:FConvexInk}) and since the limit slope of $F$ is $h(\gamma)$ (Proposition \ref{prop:FBounds}).
\end{proof}

\section{Derivation of the main theorems} \label{sec:DeriveMoments}
We can now return to the beginning of Section \ref{sec:MatrixEnumerationGeneral} and complete our proof.
Equation \ref{eqn:k'thMomentAfterMoebius} can be restated as
\begin{equation} \label{eqn:k'thMomentByDim}
\E\left((X-\E(X))^k\right) = \Theta\left(\sum_{d=0}^{k-1} G_d\right)
\end{equation}
where
\begin{equation} \label{eqn:GDefinition}
G_d = N^{-\lambda d}\sum_{\substack{V\le \bits{k}\\d(V)=d\\V\text{ robust}}}|T_V|.
\end{equation}

We need to determine which term dominates Equation \ref{eqn:k'thMomentByDim}. We use the crude upper bound of $2^{\min(d,k-d)\cdot k}$ on the number of $d$-dimensional linear subspaces $V$ of $\bits{k}$. This bound follows by considering the smaller of the two: a basis for $V$ or for its orthogonal complement.

We proceed to bound $|T_V|$ for a robust $d$-dimensional subspace $V\le \bits{k}$. When $d < \frac{k}{2}$, the trivial bound $\log|T_V|\le n\cdot h(Q_V) \le ndh(\gamma)$ suffices. Indeed, a vector sampled from $Q_V$ is determined by $d$ of its bits, each of which has entropy $h(\gamma)$. It follows that
\begin{equation} \label{eqn:G_dBoundSmalld}
G_d \le N^{d (h(\gamma)-\lambda)+\frac{kd}n}.
\end{equation}

To deal with the range $d\ge \frac{k}{2}$ we return to the notations of Equation \ref{eqn:T_VBound},
\begin{equation} \label{eqn:T_VBoundByF}
\frac{\log(|T_V|)}n \le F(m_1,\gamma) + \sum_{i=2}^{k-d}F(m_i,\gamma,\delta_i)\end{equation}
 where $m_i = |\Delta_i|$ and $\sum_{i=1}^{k-d} m_i = k$.

Lemma \ref{lem:FMonotoneInDelta} yields $F(m_i,\gamma,\delta_i) \le F(m_i,\gamma)$. By repeatedly applying Corollary \ref{cor:FConvexDiscrete} and Proposition \ref{prop:F12Bound} we get the upper bound
$$\frac{\log(|T_V|)}n \le F(2(d+1)-k,\gamma) + (k-d-1)F(2) = F(2(d+1)-k,\gamma) + (k-d-1)h(\gamma).$$

Hence, 
\begin{align*} \label{eqn:G_dBoundLarged}
\log G_d&\le -\lambda d n + dk + n(F(2(d+1)-k,\gamma)+(k-d-1)h(\gamma))\\
&=n\left(F(2(d+1)-k,\gamma)-(k-1)\lambda + (k-d-1)(h(\gamma)-\lambda)\right)+(k-d)k. \numberthis
\end{align*}

Our bounds on $G_d$ are in fact tight up to a polynomial factor in $n$ (but perhaps exponential in $k$). This follows from the existence of certain large terms in Equation \ref{eqn:GDefinition}. For $d < \frac k2$, pick any map $\varphi$ from $\{d+1,\ldots,k\}$ {\em onto} $\{1,\ldots,d\}$. Consider the space $V$ that is defined by the equations  $v_{i}=v_{\varphi(i)}$ for every $k\ge i >d$. It is clear that the space $V$ is robust. For $d \ge \frac k2$, consider the contribution of the term corresponding to
$$V = \left\{u\in \bits{k}\mid \sum_{i=1}^{t}u_i = 0 \wedge u_{t+1}=u_{t+2} \wedge u_{t+3} = u_{t+4} \wedge \ldots \wedge u_{k-1} = u_k\right\},$$
where $t = 2(d+1)-k$.

We turn to use these bounds to compute $X$'s central moments. We consider two cases, according the value of $\gamma$. 

\subsection{Moments of even order}
Let $k$ be even. By Lemma \ref{lem:FConvexInk} and Proposition \ref{prop:FBounds}, there is a positive integer $k_0=k_0(\gamma,\lambda)$ such that
$$\left\{2\le m\in \mathbb N \mid F(m,\gamma)-(m-1)\lambda > \frac m2(h(\gamma)-\lambda)\right\}=\{k_0, k_0+1, k_0+2,\ldots\}$$ 
We claim that the sum in Equation \ref{eqn:k'thMomentByDim} is dominated by either $G_{\frac k2}$ or $G_{k-1}$ depending on whether $k< k_0$ or $k\ge k_0$. 

\subsubsection{When $k<k_0$}
Since $k_0=k_0(\gamma,\lambda)$ does not depend of $n$, and since $k < k_0$ there is only a bounded number of $\bits{k}$-subspaces. We wish to compute the term $G_d=G_{\frac k2}$. We show that in this case, the sum in Equation \ref{eqn:GDefinition} is dominated by spaces of the form
\begin{equation}\label{eqn:zugiut}
V= \{v\in \bits{k} \mid v_{i_1}=v_{j_1}\wedge v_{i_2}=v_{j_2}\wedge\cdots\wedge v_{i_{\frac k2}}=v_{j_\frac k2}\},
\end{equation}
where the pairs $\{i_1,j_1\},\ldots,\{i_{\frac k2},j_{\frac k2}\}$ form a partition of $[k]$. Clearly, for such a space $V$, a matrix in $T_V$ is defined by $\frac k2$ of its rows, so
$$|T_{V}| = \binom {n}{\gamma n}^{\frac k2}.$$ 

If $U \le \bits{k}$ is robust, of dimension $\frac k2$, and not of this form \ref{eqn:zugiut}, then at least one of its associated $m_i$'s (see Equation \ref{eqn:T_VBoundByF}) equals $1$. By repeated application of Proposition \ref{prop:F12Bound}, it follows that
$$|T_U| \le N^{\frac{k}{2}F(2,\gamma)-\Omega(1)} = N^{\frac{k}{2}h(\gamma)-\Omega(1)},$$
which, as claimed, is exponentially negligible relative to $|T_V|$. The number of subspaces of the form \ref{eqn:zugiut} is $k!!$, whence
$$G_{\frac k2} = k!! \binom {n}{\gamma n}^{\frac k2} N^{-\lambda \frac k2}(1+N^{-\Omega(1)}) = N^{\frac k2 (h(\gamma)-\lambda) - \frac{k\log n}{4n}+O(\frac kn)}.$$

We turn to show that $G_d = o(G_{k/2})$ for any $d \ne \frac k2$. For $d < \frac k2$ this follows from Equation \ref{eqn:G_dBoundSmalld}. For $d > \frac k2$, due to Lemma \ref{lem:FConvexInk}, the r.h.s.\ of Equation \ref{eqn:G_dBoundLarged} is strictly convex in $d$, and therefore attains its maximum at $d=\frac k2$ or $d=k-1$. Since $k<k_0$, the former holds.\footnote{It is possible that the r.h.s.\ of Equation \ref{eqn:G_dBoundLarged} attains the same value with $d=\frac k2$ and $d=k-1$. Note that $G_{\frac k2}$ still dominates in this case, due to polynomial factors}

Equation \ref{eqn:k'thMomentByDim} yields
$$
\E\left((X-\E(X))^k\right) =  k!! \binom {n}{\gamma n}^{\frac k2} N^{-\lambda \frac k2}(1+o(1)).$$

\subsubsection{When $k\ge k_0$} 
Note that $\evens{k}$ is the one and only $(k-1)$-dimensional robust subspace of $\bits{k}$. Hence, by Equation \ref{eqn:TAsymptotics},
$$G_{k-1} = N^{-\lambda d}|T_{\evens{k}}| = N^{F(k,\gamma)-(k-1)\lambda -\frac{k\log n}{2n}+O(\frac kn)}.$$
We next show that the sum in Equation \ref{eqn:k'thMomentByDim} is dominated by this term. 
By Proposition \ref{prop:FBounds} and Equations \ref{eqn:G_dBoundLarged} and \ref{eqn:G_dBoundSmalld},
$$G_d \le N^{d(h(\gamma)-\lambda)-1+O((1-2\gamma)^k)+\frac {(k-d)k}n}$$
for all $0\le d\le k-2$. Consequently,
$$
\frac{G_d}{G_{k-1}} \le N^{(k-1-d)(\lambda-h(\gamma))+O((1-2\gamma)^k)+\frac {k\log n}{2n}+\frac {(k-d)k}n}.
$$
For large enough $k$, this is at most $N^{-\Omega(k-d)}$, so 
\begin{equation} \label{eqn:MomentLargek}
\E\left((X-\E(X))^k\right) =  G_{k-1}(1-o(1)) = N^{F(k,\gamma)-(k-1)\lambda -\frac{k\log n}{2n}+O(\frac kn)}
\end{equation}

It is left to show that Equation \ref{eqn:MomentLargek} holds for {\em all} $k \ge k_0$, but this follows again from the convexity of $F$. Namely, since $k\ge k_0$, the r.h.s.\ of Equation \ref{eqn:G_dBoundLarged} is strictly maximized by $d=k-1$, whence $G_d =o(G_{k-1})$ for $\frac k2 \le d < k-1$. For $d < \frac k2$, this inequality follows from $G_{d} < G_{\frac k2}$. 

We are now ready to state our main theorem:

\begin{theorem} \label{thm:mainEven}
For every $\gamma <\frac 12$ and $0<\lambda<h(\gamma)$ and for every even integer $k \le o(\frac n{\log n})$, the expectation $\E((X-\E(X))^k)$ is the larger of the two expressions
\begin{align*}
&k!!\binom{n}{\gamma n}^{\frac k2}N^{-\lambda\frac k2}(1+o(1)) \text{ ~~~~  and}\\
&N^{F(k,\gamma)-(k-1)\lambda-\frac{k\log n}{2n}+O(\frac kn)}.
\end{align*}
There is an integer $k_0=k_0(\gamma, \lambda)\ge 3$ such that the former term is the larger of the two when $k < k_0$ and the latter when $k\ge k_0$.
\end{theorem}

\subsection{Moments of odd order}
We turn to the case of odd $k > 2$. The arguments that we used to derive the moments of even order hold here as well, with a single difference, as we now elaborate.

The role previously held by $G_{\frac k2}$ is now be taken by either
$$G_{\frac {k-1}2} = \Theta \left(N^{\frac{k-1}{2}(h(\gamma)-\lambda)-\frac{(k-1)\log n}{4}}\right)$$ or 
$$G_{\frac {k+1}2} = \Theta \left(N^{\frac{k-3}{2}(h(\gamma)-\lambda)+F(3,\gamma)-2\lambda-\frac{(k+1)\log n}{4}}\right).$$
These asymptotics are for bounded $k$. Which of these two terms is larger depends on whether $F(3,\gamma) > (h(\gamma)-\lambda)$. This yields our main theorem for moments of odd order.

\begin{theorem} \label{thm:mainOdd}
For every $\gamma <\frac 12$ and $0<\lambda<h(\gamma)$ and for every odd integer $3 \le k \le o(\frac n{\log n})$, the expectation $\E((X-\E(X))^k)$ is the larger of the two expressions
\begin{align*}
&\Theta\left(N^{\frac{k-3}2 (h(\gamma)-\lambda)-\lambda-\frac{(k-1)\log n}4} \cdot N^{\max(h(\gamma), F(3,\gamma)-\lambda-\frac{\log n}{2n})}\right)  \text{ ~~~~  and}\\
&N^{F(k,\gamma)-(k-1)\lambda-\frac{k\log n}{2n}+O(\frac kn)}.
\end{align*} 
There is an integer $k_1=k_1(\gamma, \lambda)$ such that the former term is the larger of the two when $k < k_1$ and the latter when $k\ge k_1$.
\end{theorem}

\subsection{Normalized moments}
In this section we return to a theorem stated in the introduction. While it is somewhat weaker than our best results, we hope that is more transparent and may better convey the spirit of our main findings. 
Recall that 
$$\Var(X) = \binom{n}{\gamma n}N^{-\lambda}(1+o(1)).$$
Consider the variable $\frac{X}{\sqrt{\Var(X)}}$. By the same convexity arguments as above, its odd moments of order up to $k_0$ are $o_n(1)$. This yields the following result.

\MainNormalized*

\section{Discussion}
\subsection{Extensions and refinements} \label{subsec:Extensions}
Throughout this paper, we have limited $\gamma$ to the range $(0,\frac 12)$. What about $\gamma > \frac 12$? The function $F(k,\gamma,\delta)$ can be naturally extended to $\gamma\in(\frac 12, 1)$ and it satisfies the following obvious identity that follows by negating all bits in the underlying distribution.
$$F(m,\gamma,\delta) = \begin{cases}
F(m,1-\gamma,\delta)& \text{if }m\text{ is even}\\
F(m,1-\gamma,1-\delta)& \text{if }m\text{ is odd.}\\
\end{cases}
$$
In particular, when $\gamma > \frac 12$ and $m$ is odd, $F$ is increasing rather than decreasing in $\delta$. Also, Lemma \ref{lem:FConvexInk} is no longer valid. In fact, $F(m,\gamma)$ is larger than the linear function $m\cdot h(\gamma)-1$ when $m$ is even, but smaller than it when $m$ is odd (see Figure \ref{fig:FofGammawithDelta0And1} for an example of the odd case). 

\begin{figure}
    \centering
    \includegraphics[width=0.5\textwidth]{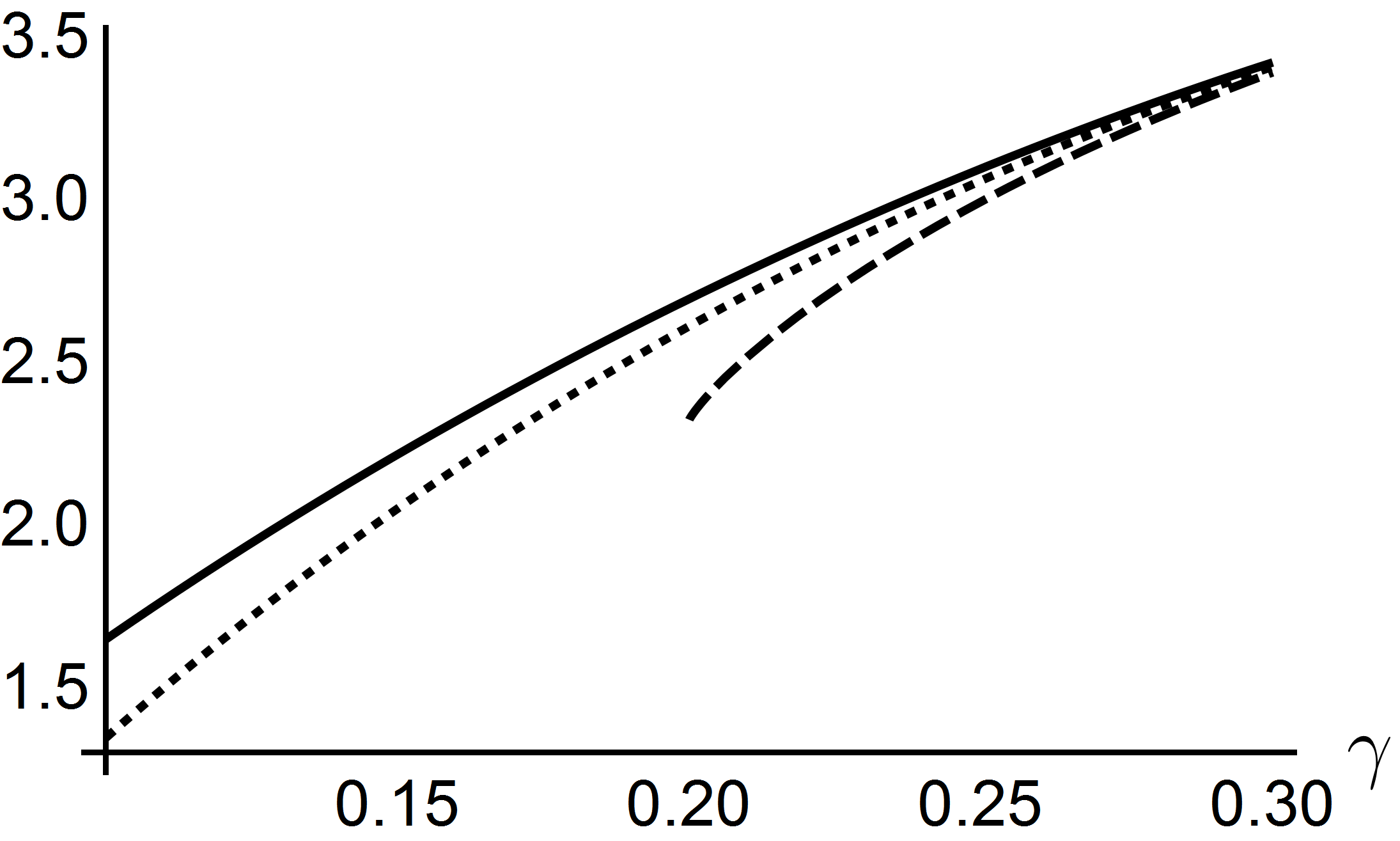}
    \caption{\label{fig:FofGammawithDelta0And1}. Illustration for Section \ref{subsec:Extensions} - Extending $F$ to $\gamma \in (\frac 12,1)$. {\bf Solid:} $F(5,\gamma)$ {\bf Dashed:} $F(5,\gamma,1)=F(5,1-\gamma)$ {\bf Dotted:} $5h(\gamma)-1$} 
\end{figure}

It can be shown that Theorem \ref{thm:mainEven} still holds in this range, but the odd moments are more complicated. The dominant term in Equation \ref{eqn:k'thMomentByDim} is no longer necessarily a product of $\evens{m}$ spaces. Rather, it may be a $(k-2)$-dimensional space, the exact parameters of which are determined by $\gamma$. 

We illustrate this unexpected additional complexity with a numerical example. Consider the following two $7$-dimensional subspaces of $\bits{9}$:
$$U = \left\{u\in\bits{9} \mid \sum_{i=1}^8 u_i = 0 \wedge u_9 = u_8\right\}$$
and
$$V = \left\{u\in \bits{9} \mid \sum_{i=1}^3 u_i = \sum_{i=4}^8 u_i = \sum_{i=7}^9 u_i\right\}.$$
For most values of $\gamma$ there holds $|T_U| > |T_V|$, but for $\gamma > 0.9997$ the opposite inequality holds.

We believe that further analysis along the lines of the present papers may yield these odd moments as well.

Similar phenomena occur when $\gamma n$ is odd. Due to parity considerations, $T_V$ is empty when there is an odd weight vector that is orthogonal to $V$. It turns out that computing the moments in this case comes down to essentially the same problem as the one described above for $\gamma > \frac 12$.

We next discuss the possible range of $k$. Namely, which moments we know. We are presently restricted to $k \le o(\frac n{\log n})$, but it is conceivable that with some additional work the same conclusions can be shown to hold for all $k \le o(n)$. The current bound arises in our analysis of the expression $\frac{G_d}{G_{k-1}}$ in Equation \ref{eqn:MomentLargek}. Our lower bound on $G_{k-1}$ includes a factor of $N^{-\frac{k\log n}{2n}}$, which is absent from our upper bound on $G_{d}$. Lemma \ref{lem:TProbability} can presumably be adapted to work for {\em general} robust subspaces, thereby improving this upper bound, thus yielding the same conclusions for $k$ up to $o(n)$. 

Pushing $k$ to the linear range $k\ge\Omega(n)$ is likely a bigger challenge, since many basic ingredients of our approach are no longer valid. If $k > (1-\lambda)n+1$, we expect our code to have dimension smaller than $k-1$, whereas our main theorems show that the $k$-th moment of $X$ is dominated by $(k-1)$-dimensional subsets of the $(\gamma n)$-th layer of $\bits{k}$.  Concretely, for $k\ge \Omega(n)$, our derivation of Equation \ref{eqn:MomentLargek} would fail, since the term $\frac{(k-d)k}n$ is no longer negligible. It is interesting to understand which terms dominate these very high moments. 

The above discussion about large $k$ is also related to the way that we sample random linear subspaces $C$ in this paper. In our model there is a negligible probability that $\dim(C)>(1-\lambda)n$. This can be avoided by opting for another natural choice, viz.\ to sample $C$ uniformly at random from among the $(1-\lambda)$-dimensional subspaces of $\bits{n}$. The effect of this choice manifests itself already in Proposition \ref{prop:CContainsSmallSet}. This effect is negligible when $d \ll (1-\lambda)n$, but becomes significant as $d$ grows, e.g., under the alternative definition $\Pr(Y\subseteq C) = 0$ whenever $\dim(Y) > (1-\lambda)n$. Presumably, $X$'s moments of order $\Theta(n)$ are sensitive to this choice of model.

There is further potential value to improving Lemma \ref{lem:TProbability}. A reduction in its error term would have interesting implications for the range $\frac{n}{\log n}\gg k > \frac{\log n}{-\log(1-2\gamma)}$. As things stand now, the difference between the upper and lower estimates in Proposition \ref{prop:FBounds} is smaller than the error term in our estimates for the moments and yields

$$N^{kh(\gamma)-1-(k-1)\lambda-\frac{k\log n}{2n}+O(\frac kn)}.$$
as our best estimate for the $k$-th moment. Reducing the error term in Lemma \ref{lem:TProbability} may significantly improve several of our results. 

\subsection{Open problems}
The long-term goal of this research is to understand the distribution of the random variable $X$. Although our computation of $X$'s moments is a step in this direction, we still do not yet have a clear view of this distribution. In particular, since all but boundedly many of $X$'s normalized moments tend to infinity, there is no obvious way to apply moment convergence theorems. 

Taking an even broader view, let us associate with a linear code $C$ the probability measure $\mu$ on $[0,1]$, with the CDF 
$$f(x) = |C|^{-1}\cdot |\{u\in C\mid \|u\| \le nx\}|.$$
We are interested in the typical behavior of this measure when $C$ is chosen at random.
In this context, our random variable $X$ corresponds to the PDF of $\mu$ at the point $\gamma$. Note that $\mu$ is typically concentrated in the range $\frac 12 \pm O(n^{-\frac 12})$, so that our questions correspond to large deviations in $\mu$. 

Many further problems concerning $\mu$ suggest themselves. What can be said about correlations between $\mu$'s PDF at two or more different points? Also, clearly, $\mu$ is binomial in expectation, but how far is it from this expectation in terms of moments, CDF, or other standard measures of similarity? We believe that the framework developed in this paper can be used to tackle these questions.

\bibliographystyle{amsplain}
\bibliography{jm}
\end{document}